\documentclass[11pt,a4]{article}
\usepackage[margin=1.5cm]{geometry}
\usepackage{amsmath,amssymb,amsthm}
\usepackage{bm}                       
\usepackage{graphicx}
\usepackage[dvipsnames]{xcolor}
\usepackage{booktabs}                 
\usepackage{microtype}
\usepackage[numbers,sort&compress]{natbib}
\usepackage[colorlinks=true,linkcolor=blue!50!black,citecolor=green!40!black,
            urlcolor=blue!50!black]{hyperref}

\newcommand{\R}{\mathbb{R}}
\newcommand{\Z}{\mathbb{Z}}
\newcommand{\E}{\mathbb{E}}
\newcommand{\QS}{\mathrm{QS}}
\newcommand{\Nrm}{\mathcal{N}}
\newcommand{\dd}{\,\mathrm{d}}
\newcommand{\Ind}[1]{\mathbf{1}\!\left[#1\right]}
\newcommand{\norm}[1]{\left\lVert #1 \right\rVert}
\newcommand{\inner}[2]{\left\langle #1,\, #2 \right\rangle}
\newcommand{\cj}[1]{\overline{#1}}
\providecommand{\tr}{\operatorname{tr}}

\newcommand{\ket}[1]{\lvert #1\rangle}
\newcommand{\bra}[1]{\langle #1\rvert}
\let\Gauge\undefined                        
\newcommand{\Gauge}{\mathcal{G}}              
\newcommand{\Proj}{P}                         
\newcommand{\Treg}[1]{\mathcal{T}_{#1}}       
\newcommand{\Zset}{\mathcal{Z}}               
\newcommand{\Lat}{\Lambda}                    
\newcommand{\dq}{d_{\Gauge}}                  
\newcommand{\cstar}{c_{\star}}                
\newcommand{\bstar}{b_{\star}}                
\newcommand{\ptgt}{p_{\mathrm{tgt}}}          
\newcommand{\aup}{\alpha_{\uparrow}}          
\newcommand{\adown}{\alpha_{\downarrow}}      

\newcommand{\vdelta}{\bm{\delta}}             
\newcommand{\vx}{\bm{x}}
\newcommand{\vy}{\bm{y}}
\newcommand{\va}{\bm{a}}                      
\newcommand{\vc}{\bm{c}}
\newcommand{\vu}{\bm{u}}                      
\newcommand{\vz}{\bm{z}}                      
\newcommand{\vw}{\bm{w}}                      
\newcommand{\ve}{\bm{e}}                      
\newcommand{\vA}{\bm{A}}                      
\newcommand{\vg}{\bm{g}}                      
\newcommand{\vv}{\bm{v}}                      
\newcommand{\vr}{\bm{r}}                      
\newcommand{\vn}{\bm{n}}                      
\newcommand{\vmu}{\bm{\mu}}                   
\newcommand{\vxi}{\bm{\xi}}
\newcommand{\phasev}{\bm{\varphi}}            
\newcommand{\noisev}{\bm{\delta}}             
\newcommand{\dispv}{\bm{\varphi}_{\mathrm{disp}}}  
\newcommand{\zerov}{0}                        

\newtheorem{theorem}{Theorem}
\newtheorem{corollary}{Corollary}
\newtheorem{proposition}{Proposition}
\newtheorem{lemma}{Lemma}

\theoremstyle{definition}
\newtheorem{definition}{Definition}

\theoremstyle{remark}
\newtheorem{remark}{Remark}
\theoremstyle{plain}
\newcommand{\restatedname}{}
\newtheorem*{restatedthm}{\restatedname}
\newenvironment{restated}[1]{\renewcommand{\restatedname}{#1}\begin{restatedthm}}{\end{restatedthm}}

\title{The Quantum Sphere: A Physically Realizable Optimization Benchmark with Provable Linear Convergence in White- and Black-Box Settings}
\author{Ofer M. Shir\\
Tel-Hai University, Israel\\
\texttt{ofersh@telhai.ac.il}}
\date{}

\begin{document}
\maketitle

\begin{abstract}
We revisit an established Quantum Control objective function, characterize its exact local metric geometry, and reposition it as a rigorous benchmark for Search and Optimization. While the trap-free topology of the landscape, established two decades ago, guarantees the existence of unhindered optimal pathways, global topology alone does not govern convergence speed. Although the so-called \emph{Quantum Sphere} is defined by the fourth power of the control field, we show that its yield gap admits an exact representation as a nonnegative linear combination of $\sin^2$ terms. Using this single identity, we establish a two-sided quadratic sandwich with ratio $4/\pi^2$ on a certified, worst-case quarter-wave region, proving the optimal peak nondegenerate modulo its affine gauge and well-conditioned.
For \emph{white-box} optimization, on a halved (eighth-wave) region, the objective satisfies a Polyak--\L{}ojasiewicz inequality with an explicit constant, granting Gradient Descent a linear rate of $1 - 16/(\pi^4\kappa_r)$, where $\kappa_r$ is the condition number of the relevant Hessian block. For \emph{black-box} optimization, we prove that the elitist single-parent Evolution Strategy converges linearly---almost surely and in expected hitting time---from a certified sublevel set, despite the landscape's continuous gauge symmetry and its lattice of periodic optima; the certified level is explicit whenever the landscape is trap-free, and near the optimum the strategy provably sees only the quadratic tangent of the peak. We corroborate the theoretical results by systematic numerical simulations. Finally, we show that additive input (actuator) noise reweights each term of the representation by a Debye--Waller factor without disrupting the landscape's structure, and outline open challenges regarding output (signal) noise, far-field plateau escape, angular-synchronization certificates, and another physically scalable benchmark (SHG with \emph{Glass}).

\end{abstract}

\section{Introduction}\label{sec:introduction}
In Search and Optimization, \emph{benchmarking} is the systematic empirical evaluation and comparison of competing algorithms to rigorously map their performance profiles. For \textbf{practitioners}, it provides a data-driven framework for algorithm selection and fine-tuning; for \textbf{theoreticians}, it serves as a vital bridge that transforms abstract concepts into practical techniques while surfacing unexpected empirical behavior. Yet, traditional benchmarks pull in two opposing directions. On one side stand \emph{synthetic} suites---e.g., quadratic forms and Rastrigin-type composites---whose analytic transparency makes them the substrate of convergence theory, but whose difficulty is an artifact of construction rather than a property of systems one would actually wish to optimize. On the other side stand \emph{real-world} objectives---simulators, learned surrogates, and physical assays \cite{SeqExpEA_TutorialGECCO2022}---that are meaningful precisely because nature or engineering defines them, yet remain opaque to analysis: lacking analytical representations, certified geometry, or provable rates. Consequently, realizable benchmark functions that are simultaneously analyzable are pressingly needed. 

Modern optimization theory provides powerful tools across both white-box and black-box paradigms. In first-order optimization, the Polyak--\L{}ojasiewicz (PL) condition \cite{lojasiewicz1963propriete,polyak1963gradient} is celebrated for ensuring linear convergence without requiring strict convexity \cite{karimi2016linear}, extending the classical paradigm of smooth, strongly convex objectives \cite{Boyd} to structured non-convex landscapes. 
Parallel to gradient-based theory, randomized search heuristics \cite{Auger2011_TRS} navigate black-box environments by relying exclusively on derivative-free objective function evaluations and their values' \textit{ranking}. 
The rigorous analysis of these methods originated with scale-invariant progress rates of Evolution Strategies (ESs) \cite{Baeck2013contemporary,ESchapter2025} on the Sphere function \cite{Rechenberg1973,Beyer2001} and general convex quadratic forms \cite{Jaegerskuepper2006}. Modern treatments have matured into sophisticated drift analyses, providing linear convergence proofs for classical mechanisms such as the 1/5$^{\text{th}}$ success rule \cite{MorinagaAkimoto2019,AAGM2022}, while Information-Geometric Optimization provides a unified theoretical framework emphasizing robustness through monotonic invariance \cite{IGO}.

In the domain of Quantum Control (QC) \cite{Koch2022}, optimization exhibits a profound geometric regularity encapsulated by the renowned Trap-Freedom results \cite{RHR2004}. 
While the conditional scope of this paradigm has sparked theoretical debate \cite{PechenTannor2011}, a sheer volume of empirical evidence in experimental systems consistently demonstrate highly efficient algorithmic trajectories \cite{Hersch92,Roslund2009} and landscape learning \cite{Shir-FOCAL}, driving inquiries into exactly why optimal controls are so ``easy to find'' \cite{HoRabitz2006}. 
Nevertheless, trap-free topologies do not imply explicit optimization rates.

To bridge this gap, this paper explores an objective function residing squarely at this intersection: physically realizable and theoretically analyzable -- the \emph{Quantum Sphere}. Instantiated via \emph{Second Harmonic Generation} (SHG)---a non-linear optical process wherein an intense electric field interacts with a non-linear medium to generate output photons possessing the combined energy of two fundamental input photons \cite{ShirBaeck2007}---this objective represents a ubiquitous calibration task for spectral-phase laser pulse shaping in QC \cite{Tannor85,Hersch92,Hersch00,Gerber07,Roslund06,LabES}. 
Combining genuine experimental relevance \cite{Shir08_QCE,Roslund2009} with an elegant, mathematically tractable formulation, we demonstrate that the Quantum Sphere admits an exact analytic representation, sharp local geometry, and non-asymptotic convergence guarantees for both white-box first-order and black-box derivative-free optimization.

\paragraph{Contributions.}
We recast a canonical QC objective function as a certified optimization benchmark, supplying the metric data that the trap-free topology omits:
\begin{enumerate}
    \renewcommand{\labelenumi}{(\roman{enumi})}
    \item \emph{Exact Representation:} An exact representation of the SHG yield gap as a nonnegative combination of $\sin^2$ terms in linear forms on the shell hypergraph (Thm.~\ref{thm:exact}), with the gauge kernel, saturation level, and far-field plateau height identified in closed form (Cors.~\ref{cor:gauge}--\ref{cor:dom});
    \item \emph{Certified Local Geometry:} A two-sided quadratic sandwich $\frac{4}{\pi^2}Q \le G \le Q$ on a certified quarter-wave region, with the lower constant attained and the upper one approached at the optimum, proving the transform-limited peak nondegenerate modulo gauge and well-conditioned (Cor.~\ref{cor:sandwich});
    \item \emph{White-box Guarantee:} Convexity and a Polyak--\L{}ojasiewicz inequality with an explicit constant on the halved (eighth-wave) region, yielding the Gradient Descent rate $1-16/(\pi^4\kappa_r)$ (Thms.~\ref{thm:whitebox-structure}--\ref{thm:whitebox-gd});
    \item \emph{Black-box Guarantee:} Linear convergence, almost surely and in expected hitting time, of the elitist $(1{+}1)$-ES with the 1/5$^{\text{th}}$ success rule from a certified sublevel set, obtained by reducing the affine gauge exactly and decomposing the sublevel sets over the wrap lattice (Lemma~\ref{lem:barrier}, Thm.~\ref{thm:blackbox}); the certified level is explicit for trap-free instances, and near the optimum the ES provably sees only the quadratic tangent (Prop.~\ref{prop:local}), whose relevant dimension is the rank $r$ rather than $N$;
    \item \emph{Empirical Validation:} Numerical corroboration across $N \in \{32, 64, 80, 160\}$, including the terminal-regime prediction and a multistart search that found no traps; a benchmark specification (Section~\ref{sec:def_qsphere}); and
    \item \emph{Noise \& Open Challenges:} The exact noise-averaged landscape under additive input (actuator) noise, alongside open challenges that culminate in two physically generated companions: a scalable benchmark family (SHG with \emph{Glass}) and a hard minimization counterpart equivalent to low-autocorrelation sequence design.
\end{enumerate}

\section{The Quantum Sphere Landscape}\label{sec:def_qsphere}
In numerical optimization, the decision variables---commonly denoted $\vx\in\R^N$---represent herein a discretized version of a continuous phase function $\varphi(\omega)$, where the $N$ coordinates correspond to individual pixels.
Then, the complex spectral field is $E(\omega)=A(\omega)\,e^{i\varphi(\omega)}$, with real, nonnegative amplitude $A$ fixed by the measured spectral intensity, $|E|^{2}=A^{2}$, and controllable spectral phase $\varphi$; $E$ is the spectrum of the control field in Eq.~\eqref{eq:tdse} below.

\begin{definition}[Quantum Sphere objective function]\label{def:qs}
The total SHG yield (under dispersion-free settings) is defined as
\begin{equation}
\begin{aligned}
\QS(\varphi) &= \int_{-\infty}^{\infty}\bigl|E_{2}(\omega)\bigr|^{2}\dd\omega, \\
\text{where } E_{2}(\omega) &= \int_{-\infty}^{\infty}E(\omega')\,E(\omega-\omega')\dd\omega' .
\end{aligned}
\label{eq:yield}
\end{equation}
\end{definition}
\paragraph{Quantum-control origin.}
In QC, the shaped pulse is the control field of a quantum system with field-free Hamiltonian $H_0$ and dipole operator $\mu$, whose state obeys the time-dependent Schr\"odinger equation
\begin{equation}
i\hbar\,\partial_t\ket{\psi(t)}=\bigl[H_0-\mu\,\mathcal{E}(t)\bigr]\ket{\psi(t)},
\qquad
\mathcal{E}(t)=\int E(\omega)\,e^{-i\omega t}\dd\omega+\mathrm{c.c.},
\label{eq:tdse}
\end{equation}
and a typical objective is a transition probability $P_{i\to f}=|\bra{f}U(T)\ket{i}|^2$, with $U$ the propagator of \eqref{eq:tdse}---the setting of the trap-free landscape results~\cite{RHR2004}. 
Nonresonant two-photon transitions establish an explicit link to Definition~\ref{def:qs}. Within second-order perturbation theory---and for intermediate levels far-detuned relative to the pulse bandwidth---the transition amplitude from state $i$ to state $f$ with two-photon frequency $\omega_{fi}$ is given by~\cite{Meshulach1998}
\begin{equation}
c^{(2)}_{f}\;\propto\;\int E(\omega')\,E(\omega_{fi}-\omega')\dd\omega'\;=\;E_2(\omega_{fi}),
\label{eq:tpa}
\end{equation}
so $P_{i\to f}\propto|E_2(\omega_{fi})|^2$ depends on the phase $\varphi$ only through the second-harmonic field of \eqref{eq:yield}. Summed over a spectrally flat band of final states with constant couplings---equivalently, detected by frequency doubling in a thin nonlinear crystal---these probabilities add up to the total yield $\QS(\varphi)$. The Quantum Sphere is thus the weak-field, two-photon limit of a QC landscape, and the standard calibration objective of pulse-shaping experiments \cite{LabES}.

\noindent The maximizer is $\varphi^\star=\text{const}$ (shaping the so-called transform-limited (TL) pulse), up to the two degrees of freedom (global phase; linear phase $=$ delay); we fix the global-phase gauge so that $\varphi^\star\equiv 0$. Throughout, $\varphi$ denotes the spectral phase and $\delta$ an additive perturbation of it. Consider the perturbed phase $\varphi=\varphi^\star+\delta$ and define the \emph{yield gap}
\begin{equation}
G(\delta)\;:=\;\QS^\star-\QS(\varphi^\star+\delta)\;\ge\;0 ,
\qquad
\QS^\star:=\QS(\varphi^\star).
\label{eq:gapdef}
\end{equation}
For each $\omega$, we define the nonnegative shell weight
\begin{equation}
\mu_\omega(\omega') := A(\omega')\,A(\omega-\omega') \ge 0 .
\end{equation}
\begin{theorem}[Exact representation; the Quantum Sphere]\label{thm:exact}
For every measurable phase perturbation $\delta$,
\begin{equation}
\boxed{%
\begin{aligned}
G(\delta) = {}&2\int\dd\omega\iint\dd\omega'\,\dd\omega''\; \mu_\omega(\omega')\,\mu_\omega(\omega'') 
 \sin^2\!\bigl(\tfrac{1}{2}\Delta_\omega[\delta](\omega',\omega'')\bigr)
\end{aligned}%
}
\label{eq:trig}
\end{equation}
where
$\Delta_\omega[\delta](\omega',\omega'')
 =\delta(\omega')+\delta(\omega-\omega')
 -\delta(\omega'')-\delta(\omega-\omega'')$
is the four-point phase difference between the two partitions
$(\omega',\,\omega-\omega')$ and $(\omega'',\,\omega-\omega'')$ of the
energy shell of total frequency $\omega$, each obeying the SHG
conservation constraint $\omega_1+\omega_2=\omega$ solved by
$\omega_2=\omega-\omega_1$.
\end{theorem}
\begin{proof}
The yield gap is $G(\delta) = \QS^\star - \int |E_2(\omega)|^2 \dd\omega$.
By writing the second-harmonic field as
$E_2(\omega)=\int \dd\omega'\,\mu_\omega(\omega')\,
 e^{i(\delta(\omega')+\delta(\omega-\omega'))}$
and expanding the intensity:
\[
|E_2(\omega)|^2
 =\iint \dd\omega'\,\dd\omega''\;
   \mu_\omega(\omega')\,\mu_\omega(\omega'')\,
   e^{i\Delta_\omega[\delta](\omega',\omega'')} .
\]
The measure $\mu_\omega(\omega')\,\mu_\omega(\omega'')$ is symmetric under
the swap $\omega'\leftrightarrow\omega''$, which sends
$\Delta_\omega\mapsto-\Delta_\omega$; the odd part therefore cancels,
leaving $\cos\Delta_\omega$ in place of $e^{i\Delta_\omega}$. Expressing
the TL peak as
$\QS^\star = \int\dd\omega\iint\dd\omega'\,\dd\omega''\,
 \mu_\omega(\omega')\,\mu_\omega(\omega'')$
and subtracting, the two integrals combine under the shared measure into
a single integrand
$\mu_\omega(\omega')\,\mu_\omega(\omega'')
 \bigl[1 - \cos\bigl(\Delta_\omega[\delta](\omega',\omega'')\bigr)\bigr]$.
Applying the half-angle identity $1 - \cos x = 2\sin^2(x/2)$ yields
Equation~\eqref{eq:trig}.
\end{proof}

\noindent \textbf{From this point on we work in the discrete formalism}---the format in which the objective is computed and realized on a physical pulse shaper. Boldface denotes $\R^N$-valued vectors, components and scalars italic; the continuum perturbation $\delta(\omega)$ of Theorem~\ref{thm:exact} becomes the decision vector $\vdelta=(\delta_0,\dots,\delta_{N-1})\in\R^N$ sampling $\delta$ on the pixel grid, at fixed amplitude $\vA\in\R_{\ge0}^N$ and TL gauge $\varphi^\star=0$. 
Every statement below admits a verbatim continuum analogue (each per-shell sum $\sum_{i\in S_s}$ becoming an integral $\int\dd\omega'$ against $\mu_\omega(\omega')$, and $\sum_s$ becoming $\int\dd\omega$); no proof relies on discreteness. For each $s\in\{0,\dots,2N-2\}$, the \emph{shell}
\begin{equation}
S_s \;=\; \{\, i \mid 0\le i\le N-1,\; 0\le s-i\le N-1 \,\}
\label{eq:shell}
\end{equation}
comprises index pairs summing to $s$. 
The objective is
\begin{equation}
Y(\vdelta)\;=\;\sum_{s=0}^{2N-2}\;\Bigl|\,
\sum_{i\in S_s} A_i\,A_{s-i}\; e^{\,\mathrm{i}(\delta_i+\delta_{s-i})}\Bigr|^{2},
\label{eq:yield-disc}
\end{equation}
maximized at $\vdelta=0$ (non-uniquely; see Cor.~\ref{cor:gauge}), with yield gap $G(\vdelta)=Y(0)-Y(\vdelta)\ge0$.

\begin{corollary}[Discrete Representation]\label{cor:sin2-disc}
For every $\vdelta\in\R^{N}$,
\begin{equation}
\label{eq:sin2-sum}
\begin{aligned}
G(\vdelta) = 2\sum_{s=0}^{2N-2} & \sum_{i,p\,\in S_s} A_i A_{s-i}\,A_p A_{s-p}  \sin^{2}\!\Bigl( \tfrac{\delta_i+\delta_{s-i}-\delta_p-\delta_{s-p}}{2} \Bigr),
\end{aligned}
\end{equation}
where the inner sum runs over all ordered pairs $(i,p) \in S_s \times S_s$,
with per-shell phases $v^{(s)}_i=\delta_i+\delta_{s-i}$ and weights
$\alpha^{(s)}_i=A_iA_{s-i}$.
\end{corollary}
\begin{proof}[Proof sketch]
Expand each shell's coherent sum via the identity 
\[ \bigl(\sum_i\alpha^{(s)}_i\bigr)^2-\bigl|\sum_i\alpha^{(s)}_i e^{\mathrm{i}v^{(s)}_i}\bigr|^2 = 2\sum_{i,p}\alpha^{(s)}_i\alpha^{(s)}_p\sin^2\!\bigl(\tfrac{v^{(s)}_i-v^{(s)}_p}{2}\bigr);\] summing over $s$ gives \eqref{eq:sin2-sum}. See Appendix~\ref{s:thm12}.
\end{proof}
Maximizing the per-shell coherent sum $\sum_{i\in S_s}\alpha^{(s)}_i e^{\mathrm{i}v^{(s)}_i}$ is an instance of \emph{Angular Synchronization} \cite{SINGER2011}, which is generally a non-convex problem. 

\begin{corollary}[Gauge Zero-Set and Hessian Kernel]\label{cor:gauge}
Suppose the support of $\vA$ is a contiguous index window $W$ with $|W|\ge3$. Then $G(\vdelta)=0$ if and only if $\vdelta$ is affine \emph{modulo $2\pi$} on $W$: there exist $a,b\in\R$ such that $\delta_i\equiv a+b\,i \pmod{2\pi}$ for all $i\in W$. 
\end{corollary}
\noindent The kernel of the Hessian $H=\nabla^2G(0)$ restricted to $W$ is exactly the two-dimensional space of affine phases; we write $\Gauge:=\ker H$ for this \emph{gauge} subspace, spanned by the constant vector $\mathbf 1=(1,\dots,1)$ and the linear ramp $\vr=(0,1,\dots,N-1)$.
\begin{proof}[Proof sketch]
$G(\vdelta)=0$ forces $v^{(s)}_i=\delta_i+\delta_{s-i}$ constant in $i$ (mod $2\pi$) on each shell; the even shell $s=2i$ pairs $(i,i)$ against $(i-1,i+1)$, forcing the discrete second difference $\delta_{i-1}-2\delta_i+\delta_{i+1}\equiv0\pmod{2\pi}$ at every interior index of $W$ (whence $|W|\ge3$), i.e.\ affinity $\delta_i\equiv a+b\,i$; the converse is invariance under global phase and linear delay. For the kernel, the analytic SHG Hessian was first derived in closed form by \cite{Shir-FOCAL} for the \emph{yield} $J_{\mathrm{SHG}}=Y$ at $\phi^\star=0$, so $H=\nabla^2G(0)=-\mathbf{H}_{\mathrm{SHG}}$ up to the objective's overall sign;\footnote{That work specializes to a Gaussian amplitude on a truncated window; our representation is amplitude-general and discrete, and the null-space characterization is not treated there.} the same second-difference argument over $\R$ recovers the 2D affine kernel. See Appendix~\ref{s:thm12}.
\end{proof}

\paragraph{Zero set: gauge and wrap lattice.}
By Corollary~\ref{cor:gauge}, the zero set of $G$ is
\begin{equation}
\Zset\;=\;2\pi\Z^N+\Gauge ,
\label{eq:zeroset}
\end{equation}
the union of the two symmetries of the objective. The affine gauge $\Gauge$ is a continuous 2D symmetry and is \emph{local}: it is the kernel of $H$. The per-pixel wrap lattice $2\pi\Z^N$ is a discrete symmetry and is \emph{global}: $Y$ depends on $\vdelta$ only through $e^{\mathrm{i}\vdelta}$. We remove the gauge by working modulo $\Gauge$, with the orthogonal projector $P$ onto $\Gauge^\perp$ and the gauge-quotient norm
\begin{equation}
\dq(\vdelta)\;:=\;\min_{\vg\in\Gauge}\norm{\vdelta-\vg}_\infty ;
\label{eq:dq}
\end{equation}
the wrap lattice then survives as a discrete lattice of optima in the quotient (Appendix~\ref{s:blackbox}). 

\begin{definition}[Weighted $\sin^2$ class]\label{def:class}
Given weights $w_k\ge0$ and vectors $\va_k\in\R^N$, $k\in K$ finite, let
\begin{equation}
G(\vdelta)\;=\;\sum_{k\in K} w_k\,
\sin^{2}\!\bigl(\inner{\va_k}{\vdelta}\bigr).
\label{eq:class}
\end{equation}
The Quantum Sphere is an instance of this class, with $k=(s,i,p)$ ranging over
ordered shell pairs and $\ve_i$ the $i$-th standard basis vector of $\R^N$,
\[
\va_k=\tfrac12\bigl(\ve_i+\ve_{s-i}-\ve_p-\ve_{s-p}\bigr),
\qquad
w_k = 2\,A_iA_{s-i}A_pA_{s-p},
\]
matching \eqref{eq:sin2-sum} exactly. Both the forms $\va_k$ and the weights $w_k$ are those of the gap expanded about the TL peak $\varphi^\star=0$; the representation then holds at every $\vdelta$, but the constants $w_k$ are fixed by that reference.\\ 
The quadratic tangent and the Hessian at the optimum are
\begin{equation}
Q(\vdelta)=\sum_{k\in K} w_k\inner{\va_k}{\vdelta}^{2}
=\tfrac12\,\vdelta^{\!\top}H\vdelta,
\quad
H=\sum_{k\in K}2\,w_k\,\va_k^{\vphantom{\top}} \va_k^{\!\top}.
\label{eq:QH}
\end{equation}
\end{definition}
\begin{corollary}[Global domination and saturation]\label{cor:dom}
Define the \emph{incoherent yield}
\begin{equation}
Y_{\mathrm{inc}}
\;:=\;
2\sum_{s}\sum_{i\in S_s}\bigl(\alpha^{(s)}_i\bigr)^{2}
\;-\;\sum_{s\ \mathrm{even}}\bigl(\alpha^{(s)}_{s/2}\bigr)^{2}.
\label{eq:yinc}
\end{equation}
Then, for every $\vdelta\in\R^{N}$,
\begin{equation}
G(\vdelta)\;\le\;Q(\vdelta)
\label{eq:dom}
\end{equation}
and
\begin{equation}
G(\vdelta)\;\le\;2\bigl(Y(0)-Y_{\mathrm{inc}}\bigr);
\label{eq:sat}
\end{equation}
moreover, when $\vdelta$'s coordinates are i.i.d.\ uniform on $[-\pi,\pi]$,
\begin{equation}
\E\bigl[G(\vdelta)\bigr]\;=\;Y(0)-Y_{\mathrm{inc}}
\qquad\text{exactly,}
\label{eq:avg}
\end{equation}
the \emph{far-field plateau height}.
\end{corollary}
\begin{proof}[Proof sketch]
Writing the pair index $k=(s,i,p)$ with weight $w_k=2\alpha^{(s)}_i\alpha^{(s)}_p$ and half-difference $\ell_k=\tfrac12(v^{(s)}_i-v^{(s)}_p)$, the representation \eqref{eq:sin2-sum} reads $G=\sum_{k\in K}w_k\sin^2\ell_k$ with tangent $Q(\vdelta)=\sum_k w_k\ell_k^2$. Then \eqref{eq:dom} is $\sin^2\ell_k\le\ell_k^2$; the reversal symmetry $v^{(s)}_{s-i}=v^{(s)}_i$ renders the pairs $p\in\{i,s-i\}$ phase-blind, of total weight $2Y_{\mathrm{inc}}$, so bounding $\sin^2\le1$ on the responsive remainder gives \eqref{eq:sat}; and under i.i.d.\ $\mathcal{U}[-\pi,\pi]$ each responsive term has $\E[\sin^2\ell_k]=\tfrac12$, yielding \eqref{eq:avg}.
See Appendix~\ref{s:thm12}.
\end{proof}


\begin{corollary}[Quarter-wave sandwich]\label{cor:sandwich}
For $\theta\in(0,1]$ define the gauge-fixed region family
\begin{equation}
\Treg{\theta}\;:=\;\Bigl\{\vdelta:\ \dq(\vdelta)\le\tfrac{\theta\pi}{4}\Bigr\},
\label{eq:regions}
\end{equation}
so that $\vdelta\in\Treg{\theta}$ implies
$|\inner{\va_k}{\vdelta}|\le\theta\pi/2$ for all $k$. Then on the
quarter-wave region $\Treg{1}$,
\begin{equation}
\boxed{\;\tfrac{4}{\pi^{2}}\,Q(\vdelta)
\;\le\;G(\vdelta)\;\le\;Q(\vdelta),\;}
\label{eq:sandwich}
\end{equation}
the upper bound holding globally by \eqref{eq:dom}.
\end{corollary}

\begin{proof}[Proof sketch]
On $\Treg{1}$ every argument satisfies
$|\inner{\va_k}{\vdelta}|\le\pi/2$ (the $\va_k$ annihilate $\Gauge$, so the
gauge minimization in \eqref{eq:regions} is free), where Jordan's
inequality $\sin^2x\ge(4/\pi^2)x^2$ applies termwise; the upper bound is
Corollary~\ref{cor:dom}. See Appendix~\ref{s:thm12}.
\end{proof}
\paragraph{Convexity status.}
No global convexity notion can apply to $G$: the gap is bounded \eqref{eq:sat} and $2\pi$-periodic per coordinate, and a nonconstant function with either property is never convex on $\R^N$.\footnote{On the torus---the natural domain of spectral phase variables---geodesically convex functions are necessarily constant.} The suitable treatment is therefore inequality-based, as in \eqref{eq:sandwich}; convexity itself is confined to the halved region $\Treg{1/2}$, sharply (Theorem~\ref{thm:whitebox-structure}, witnessed by the so-called \textit{comb} of Fig.~\ref{fig:comb}).


\paragraph{Benchmark specification.}
To make the instance reusable, we fix the following protocol.
(i)~\emph{Objective:} the yield \eqref{eq:yield-disc}, evaluated by \emph{linear} self-convolution, to be maximized---equivalently, the gap $G=Y(0)-Y\ge0$, to be minimized. The circular (length-$N$ DFT) evaluation defines a different landscape, in which the delay is not a symmetry; we use it only as a cross-check.
(ii)~\emph{Amplitudes:} the top hat $\vA=\mathbf 1$ by default; Gaussian and random profiles as variants.
(iii)~\emph{Domain:} all of $\R^N$, unconstrained; the wrap lattice is part of the landscape.
(iv)~\emph{Performance measure:} the relative gap $G/Y(0)$. It is gauge- and lattice-invariant; distances to a particular optimizer are not meaningful, since the optimum is the set $\Zset$.
(v)~\emph{Dimension:} report both $N$ and the number of active directions $r=\operatorname{rank}H$ ($=N-2$ for full support), which is the dimension that governs the rates below.
(vi)~\emph{Initialization:} near field---a random direction scaled to $\norm{P\vdelta^0}_\infty=\pi/8$ (Section~\ref{sec:numerical}); far field---i.i.d.\ uniform phases on $[-\pi,\pi]$, where $\E[G]=Y(0)-Y_{\mathrm{inc}}$ (Corollary~\ref{cor:dom}).


\section{Linear Convergence Results}\label{sec:linear_convergence}
This section establishes the linear convergence of Gradient Descent (white-box) and of an ES (black-box) on the Quantum Sphere. 
We begin with the white-box results, which yield explicit constants; their core ingredients---star-convexity, global smoothness, and the absence of critical points inside $\Treg{1}$---subsequently underpin the black-box analysis.
Proofs are provided in Appendices~\ref{s:aux}--\ref{s:comb}.

\paragraph{Standing setting.}
$G(\vdelta)=\sum_k w_k\sin^2\inner{\va_k}{\vdelta}$ is the Quantum Sphere
gap of Theorem~\ref{thm:exact}, in the class form of
Definition~\ref{def:class} (weights $w_k\ge0$; half four-point vectors
$\va_k$, so $|\inner{\va_k}{\vdelta}|\le2\norm{\vdelta}_\infty$).
$H=\nabla^2G(0)$ of \eqref{eq:QH} has
rank $r$, extreme nonzero eigenvalues $\lambda_1\ge\lambda_r>0$,
condition number $\kappa_r=\lambda_1/\lambda_r$, and null space
$\Gauge$ (the affine gauge, Cor.~\ref{cor:gauge}); $P$ projects onto
$\Gauge^\perp=\operatorname{range}H$, all distances are measured through $\norm{P\vdelta}_2$, and
$Q(\vdelta)=\tfrac12\vdelta^{\!\top}\!H\vdelta\ge\tfrac{\lambda_r}{2}\norm{P\vdelta}_2^2$. The regions
$\Treg{\theta}$ are those of \eqref{eq:regions}, with membership
certified by $\norm{P\vdelta}_\infty\le\theta\pi/4$; $\Treg{1}$ is the
\emph{quarter-wave} region, $\Treg{1/2}$ the \emph{eighth-wave} region, and $\Zset=2\pi\Z^N+\Gauge$ is the zero set.
The two-sided sandwich \eqref{eq:sandwich} on $\Treg{1}$ is the
structural input to everything below.

\paragraph{White-box: convexity, PL, and Gradient Descent.}
Gradient methods consume the derivative $\nabla G=\sum_kw_k\sin(2\inner{\va_k}{\vdelta})\,\va_k$, whose curvature terms $\cos(2\inner{\va_k}{\vdelta})$ keep their sign only while $|\inner{\va_k}{\vdelta}|\le\pi/4$, i.e., on $\Treg{1/2}$.

\begin{theorem}[White-box: convexity and PL at half radius]
\label{thm:whitebox-structure}
On $\Treg{1/2}$ the gap is convex and satisfies the
Polyak--{\L}ojasiewicz inequality
\[
\norm{\nabla G(\vdelta)}_2^{\,2}\;\ge\;\frac{32\,\lambda_r}{\pi^4}\,
G(\vdelta);
\]
on $\Treg{\theta}$ with $\theta<\tfrac12$ it is
$\cos(\theta\pi)\lambda_r$-strongly convex modulo gauge.
Both radii are sharp, already for the Quantum Sphere itself: the alternating
comb renders the Hessian indefinite on $\mathcal T_\theta$ for every
$\theta>\tfrac12$, so convexity fails immediately beyond the half radius; and
the same comb, at the quarter-wave radius, is a critical point with positive
gap on $\partial\mathcal T_1$, so no star-convexity or PL constant survives at
the unit radius.
\end{theorem}

\begin{theorem}[White-box: linear convergence of the Gradient Descent]
\label{thm:whitebox-gd}
Let $\vdelta^{t+1}=\vdelta^t-\lambda_1^{-1}\nabla G(\vdelta^t)$ with
$\norm{P\vdelta^0}_2\le\tfrac{1}{8\sqrt{\kappa_r}}$. Then for all $t$,
\[
G(\vdelta^t)\;\le\;\Bigl(1-\frac{16}{\pi^4\kappa_r}\Bigr)^{\!t}
G(\vdelta^0),
\]
the iterates converge to a single point of the TL gauge orbit, and the gauge coordinates are constants of the motion.
\end{theorem}

\noindent
Theorem~\ref{thm:whitebox-structure} is a property of the \emph{region} $\Treg{1/2}$, whereas Theorem~\ref{thm:whitebox-gd} is a \emph{trajectory} guarantee from the smaller initialization set $\norm{P\vdelta^0}_2\le1/(8\sqrt{\kappa_r})$, which implies $G(\vdelta^0)\le\lambda_r/128$ (Lemma~\ref{s:lem:bridge}). Beyond convexity, the PL inequality extends to every $\Treg{\theta}$ with $\theta<1$, with a constant degrading to $0$ at the quarter-wave boundary (Corollary~\ref{s:cor:gradedpl}); hence $G$ has no critical points in the open quarter-wave region other than its zeros. This property will be used in the black-box analysis below, rather than convexity.

\paragraph{Black-box: the algorithm.}
We analyze the elitist isotropic $(1{+}1)$-ES with success-based step-size adaptation in the implementation of \citet{Kern2004}, as analyzed by \citet{AAGM2022}: their Algorithm~2.1 with the covariance matrix fixed to the identity ($\Sigma_t\equiv I$, $\kappa=1$), applied to $f=G$ on $\R^N$. 
Starting from $\vdelta^0\in\R^N$ and $\sigma_0>0$, iteration $t$ (i)~samples $\vz_t\sim\Nrm(0,I_N)$; (ii)~generates the candidate $\vdelta'=\vdelta^t+\sigma_t\vz_t$; and (iii)~if $G(\vdelta')\le G(\vdelta^t)$, sets $\vdelta^{t+1}=\vdelta'$ and $\sigma_{t+1}=\aup\sigma_t$, and otherwise keeps $\vdelta^{t+1}=\vdelta^t$ and sets $\sigma_{t+1}=\adown\sigma_t$. The factors $\aup>1>\adown>0$ encode the target success probability $\ptgt:=\ln(1/\adown)/\ln(\aup/\adown)$; the 1/5$^{\text{th}}$ success rule is $\ptgt=1/5$, i.e., $\adown=\aup^{-1/4}$ (Section~\ref{sec:numerical} uses $\aup=e^{0.1}$, $\adown=e^{-0.025}$).

\paragraph{Two symmetries, two difficulties.}
The ES runs in $\R^N$, where $G$ has two symmetries that the classical Sphere function lacks. The \emph{affine gauge} is harmless: since $G(\vdelta)=G(P\vdelta)$ and $P\vz$ is independent of $(I-P)\vz$, the projected incumbent $P\vdelta^t$ is itself a $(1{+}1)$-ES on the $r$-dimensional quotient $\Gauge^\perp$, driven by the same acceptance events, while the gauge coordinates perform a neutral random walk with steps $\sigma_t(I-P)\vz_t$ that selection never sees (Lemma~\ref{s:lem:proj}).
The \emph{wrap lattice} $2\pi\Z^N$ is not harmless: every sublevel set $\{G\le c\}$ with $c>0$ is an infinite periodic union, so no sublevel set of positive level lies inside $\Treg{1}$, and a mutation may jump into a periodic copy of the basin. Containment must therefore be measured \emph{modulo the lattice}, which leads to the \emph{quarter-wave barrier}
\begin{equation}
\cstar\;:=\;\min\bigl\{G(\vdelta):\ \operatorname{dist}(\vdelta,\Zset)\ge\tfrac{\pi}{4}\bigr\},
\qquad
\operatorname{dist}(\vdelta,\Zset):=\min_{\zeta\in\Zset}\norm{\vdelta-\zeta}_\infty ,
\label{eq:barrier}
\end{equation}
the lowest gap outside the lattice translates of the quarter-wave region. Call a local minimizer $\vdelta$ of $G$ with $G(\vdelta)>0$ \emph{spurious} (a trap of the yield), and let $G_{\mathrm{trap}}$ be the infimum of $G$ over spurious local minimizers ($G_{\mathrm{trap}}:=+\infty$ if there are none).

\begin{lemma}[Quarter-wave barrier]\label{lem:barrier}
The barrier is positive and bounded by two explicit witnesses, the quarter-wave comb $\tfrac{\pi}{4}\vu$ (Fig.~\ref{fig:comb}) and the single-pixel half-wrap $\pi\ve_j$:
\[
\min\Bigl\{\tfrac{\lambda_r}{8},\,G_{\mathrm{trap}}\Bigr\}\;\le\;\cstar\;\le\;\min\Bigl\{\tfrac{C}{8},\ \min_{j\in W}G(\pi\ve_j)\Bigr\},
\qquad C=\vu^{\!\top}\!H\vu .
\]
For every $c<\cstar$, the sublevel set $\{G\le c\}$ is, modulo the gauge, the disjoint union of the lattice translates of a single star-shaped basin that lies in the interior of $\Treg{1}$.
\end{lemma}

\noindent
The lower bound is where topology meets metric geometry: in the absence of traps, the lowest point outside the quarter-wave tubes must lie on their boundary, where the sandwich forces $G\ge\lambda_r/8$.

\begin{theorem}[Black-box: linear convergence of the $(1{+}1)$-ES]\label{thm:blackbox}
Let $\ptgt<\tfrac12$ and
\[
\bstar\;:=\;\min\Bigl\{G_{\mathrm{trap}},\ \frac{\lambda_r}{128}\,\ptgt^{\,2/r}\Bigr\}.
\]
Fix a level $b<\bstar$. There are constants $B>0$ and $C_0,C_1<\infty$, depending on the instance, on $b$ and on $(\aup,\adown)$ but not on $\vdelta^0$, $\sigma_0$ or $\varepsilon$, such that for every initial state with $G(\vdelta^0)\le b$ and every $\sigma_0>0$:
\begin{enumerate}
\renewcommand{\labelenumi}{(\roman{enumi})}
\item \emph{Confinement modulo the lattice:} $G(\vdelta^t)$ is nonincreasing, and every incumbent lies in a lattice translate of the quarter-wave region, $\vdelta^t\in2\pi\Z^N+\Treg{1}$;
\item \emph{Expected hitting time:} for every $\varepsilon\in(0,G(\vdelta^0)]$, the first time $T_\varepsilon$ with $G(\vdelta^t)\le\varepsilon$ satisfies
\[
\E[T_\varepsilon]\;\le\;\frac{1}{2B}\ln\frac{G(\vdelta^0)}{\varepsilon}\;+\;C_1\Bigl|\ln\frac{\sigma_0}{\operatorname{dist}_2(\vdelta^0,\Zset)}\Bigr|\;+\;C_0 ,
\]
where $\operatorname{dist}_2(\vdelta^0,\Zset)=\min_{\vn\in\Z^N}\norm{P(\vdelta^0-2\pi\vn)}_2$ is the Euclidean distance from $\vdelta^0$ to the set of global minimizers;
\item \emph{Almost-sure linear convergence:} $\displaystyle\limsup_{t\to\infty}\tfrac1t\ln G(\vdelta^t)\le-2B$ almost surely.
\end{enumerate}
\end{theorem}

\noindent
The proof (Appendix~\ref{s:blackbox}) reduces the ES exactly to the quotient $\Gauge^\perp$, decomposes the sublevel sets into basin copies (Lemma~\ref{lem:barrier}), and verifies the objective-function conditions of the drift analysis of \citet[Sect.~4]{AAGM2022}---inner and outer balls from the sandwich, and uniform improvability from the absence of critical points inside $\Treg{1}$ noted above---with one modification: the outer ball holds only per basin copy, and the copies add at most $(128b/\lambda_r)^{r/2}<\ptgt$ to the success probability of an oversized step (Lemma~\ref{s:lem:images}).

\paragraph{Scope.}
Theorem~\ref{thm:blackbox} is qualitative in its rate: $B$ is produced by a drift argument whose constants are not tight, and we claim no Sphere-relative factor. Its level condition has two parts, and is of the same order as the initialization level $\lambda_r/128$ of Theorem~\ref{thm:whitebox-gd}. The factor $\tfrac{\lambda_r}{128}\ptgt^{2/r}$ stems from a union bound over periodic images and is not expected to be sharp; for the top-hat instances of Section~\ref{sec:numerical} it corresponds to relative gaps $G/Y(0)$ between $1.4\times10^{-4}$ and $6.8\times10^{-4}$, and the ES converges linearly from far outside this set (Section~\ref{sec:numerical}). The factor $G_{\mathrm{trap}}$ is where the trap-free property enters: we do not prove that the discrete Quantum Sphere is trap-free, but in the absence of traps $\bstar=\tfrac{\lambda_r}{128}\ptgt^{2/r}$ is explicit, and no spurious local minimizer was found numerically (Section~\ref{sec:numerical}).

\paragraph{The terminal regime.}
Close to the optimum, the ES cannot distinguish the Quantum Sphere from its quadratic tangent:

\begin{proposition}[Terminal regime]\label{prop:local}
Let $\vy\in\Gauge^\perp\setminus\{0\}$, write $\vy=\rho\,\hat{\vy}$ with $\norm{\hat{\vy}}_2=1$, and let $\sigma=s\rho$. As $\rho\to0$, the probability that $G$ and $Q$ disagree on accepting the same candidate,
\[
\Pr_{\vz}\Bigl[\,\Ind{G(\vy+\sigma\vz)\le G(\vy)}\neq\Ind{Q(\vy+\sigma\vz)\le Q(\vy)}\,\Bigr],
\]
tends to $0$, uniformly in $\hat{\vy}$ and in $s$ on compact subsets of $(0,\infty)$. Consequently, in the scale-invariant coordinates $(\vy/\rho,\,\sigma/\rho)$, the one-step transition law of the $(1{+}1)$-ES on $G$ converges in total variation to that of the $(1{+}1)$-ES on the quadratic tangent $Q$.
\end{proposition}

\noindent
Proposition~\ref{prop:local} identifies the terminal dynamics with those of the ES on the ellipsoid $Q$ restricted to $\Gauge^\perp$, whose convergence rate is governed by $\lambda_r/\tr(H)$ \citep{Morinaga2021}; the relevant dimension is the rank $r$, not the nominal $N$. If $H$ were isotropic on $\Gauge^\perp$, $Q$ would be a Sphere in $r$ dimensions and the classical progress-rate constants would apply (Lemma~\ref{s:lem:rate}); the SHG instances considered here are not isotropic (for the top-hat instances of Section~\ref{sec:numerical}, $\kappa_r=1.33$--$1.49$), but numerically the effect of this mild anisotropy on the rate is small (Section~\ref{sec:numerical}). Proposition~\ref{prop:local} is a one-step statement: transferring the \emph{rate} from $Q$ to $G$ would require continuity of the stationary behavior of the normalized chain, which we do not prove. Numerically the two rates coincide (Table~\ref{tab:es-validation}).

\begin{figure}[h!]
\centering
\includegraphics[width=\columnwidth]{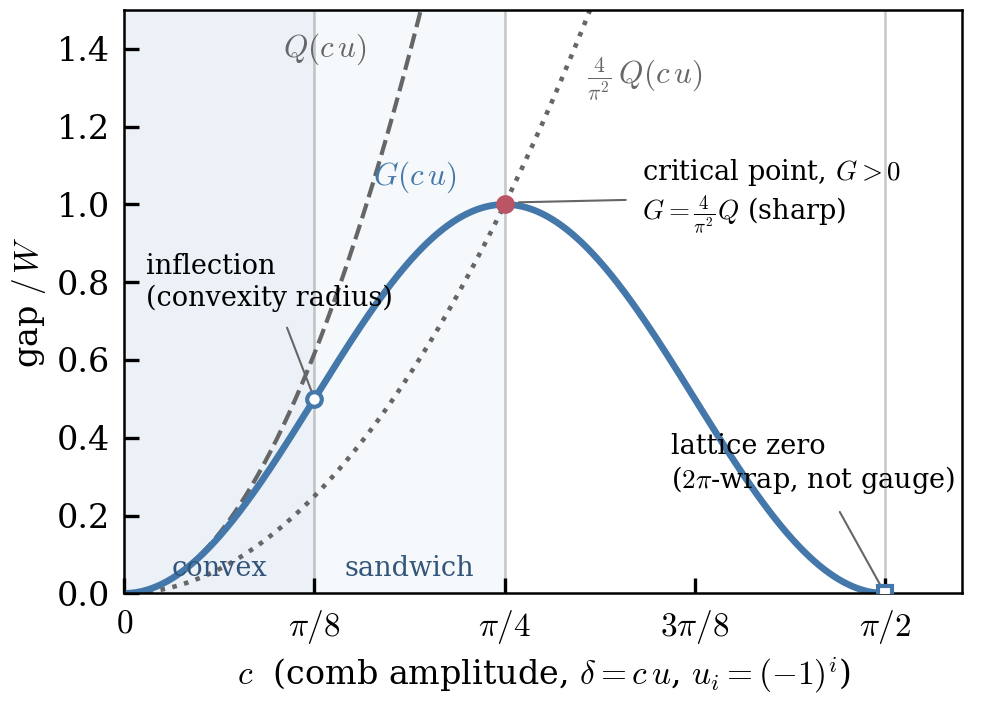}
\caption{\textbf{1D landscape section along the alternating comb} $\vu = ((-1)^i)_{i=0}^{N-1}$ (normalized by $C/8$, where $C = \vu^{\!\top}H\vu > 0$), witnessing the $1{:}2{:}4$ radius hierarchy. The exact gap reduces to $\sin^2(2c)$ (solid blue), invariant to the amplitude profile $\vA$; it lies below the quadratic tangent $Q$ (dashed) everywhere, and above the lower bound $\frac{4}{\pi^2}Q$ (dotted) on the certified region $c\le\pi/4$ only---the touch at $c=\pi/4$ is where the sandwich goes tight, and it breaks beyond. Three milestones:
(i)~$c = \pi/8$ (open circle): inflection point marking the boundary of convexity (Thm.~\ref{thm:whitebox-structure});
(ii)~$c = \pi/4$ (red dot): a critical point with positive gap on $\partial\Treg{1}$, where the sandwich constant $4/\pi^2$ is attained; it caps the quarter-wave barrier, $\cstar\le C/8$ (Lemma~\ref{lem:barrier}); and
(iii)~$c = \pi/2$ (open square): a zero of $G$ outside the gauge---a wrapped affine phase of slope $\pi$, i.e., a point of the wrap lattice.}
\label{fig:comb_section}
\label{fig:comb}
\end{figure}
\begin{table*}[t]
\centering
\small
\setlength{\tabcolsep}{6pt}
\renewcommand{\arraystretch}{1.25}
\begin{tabular}{@{}p{0.15\textwidth}p{0.38\textwidth}p{0.38\textwidth}@{}}
\toprule
 & \textbf{Gradient Descent} \,(white-box) & \textbf{$(1{+}1)$-ES} \,(black-box) \\
\midrule
Landscape region
  & convexity and PL on $\Treg{1/2}$ (eighth-wave)
    \hfill{\footnotesize[Thm.~\ref{thm:whitebox-structure}]}
  & sandwich, star shape and no critical points on $\Treg{1}$ (quarter-wave)
    \hfill{\footnotesize[Cor.~\ref{cor:sandwich}, Lem.~\ref{lem:barrier}]} \\
Certified start
  & $\norm{P\vdelta^0}_2\le\tfrac{1}{8\sqrt{\kappa_r}}$, hence $G(\vdelta^0)\le\tfrac{\lambda_r}{128}$
    \hfill{\footnotesize[Thm.~\ref{thm:whitebox-gd}]}
  & any $\vdelta^0$ with $G(\vdelta^0)<\bstar=\min\{G_{\mathrm{trap}},\tfrac{\lambda_r}{128}\ptgt^{2/r}\}$
    \hfill{\footnotesize[Thm.~\ref{thm:blackbox}]} \\
\quad\emph{--- sharp?}
  & radii yes: the comb pins the convexity radius; PL fails at $\partial\Treg{1}$
    \hfill{\footnotesize[Props.~\ref{cx:prop},\,\ref{cx:crit}]}
  & no; the barrier obeys $\cstar\le\min\{\tfrac{C}{8},\min_jG(\pi\ve_j)\}$
    \hfill{\footnotesize[Lem.~\ref{lem:barrier}]} \\
\addlinespace[3pt]
Rate (per step)
  & $\displaystyle 1-\frac{16}{\pi^{4}\kappa_r}$ \;{\footnotesize(contraction factor)}
    \hfill{\footnotesize[Thm.~\ref{thm:whitebox-gd}]}
  & linear, a.s.\ and in expected hitting time; constant not explicit
    \hfill{\footnotesize[Thm.~\ref{thm:blackbox}]} \\
\addlinespace[3pt]
Terminal regime
  & strong cvx.\ modulo gauge on $\Treg{\theta}$, $\theta<\tfrac12$:
    modulus $\cos(\theta\pi)\lambda_r$
    \hfill{\footnotesize[Thm.~\ref{thm:whitebox-structure}]}
  & the ES on the quadratic tangent $Q$ \hfill{\footnotesize[Prop.~\ref{prop:local}]}; measured $0.88$--$0.96$ of the Sphere reference $2\psi^\star_{1/5}/r$
    \hfill{\footnotesize[Tab.~\ref{tab:es-validation}]} \\
\addlinespace[3pt]
Gauge coordinates
  & frozen --- constants of the motion
    \hfill{\footnotesize[Thm.~\ref{thm:whitebox-gd}]}
  & neutral martingale; freezes as $\sigma_t\to0$
    \hfill{\footnotesize[Lem.~\ref{s:lem:proj}]} \\
Wrap lattice
  & never leaves its basin
    \hfill{\footnotesize[Thm.~\ref{thm:whitebox-gd}]}
  & may hop between lattice copies of the basin
    \hfill{\footnotesize[Thm.~\ref{thm:blackbox}(i)]} \\
Invariances
  & none beyond the affine gauge
  & monotone transforms of $G$ (comparison-based selection) \\
Oracle cost / step
  & full model gradient
  & one comparison \\
\bottomrule
\end{tabular}
\caption{The Quantum Sphere's breakdown of results: white-box (gradient) versus black-box (ES) guarantees.}
\label{tab:ledger}
\end{table*}

\paragraph{Summarizing the theoretical results.}
Table~\ref{tab:ledger} provides the breakdown of results. Both families rest on the same two facts---the shell representation and the sandwich of Cor.~\ref{cor:sandwich}---but consume them differently.
Gradient descent consumes the derivative $\nabla G\propto\sin\!\bigl(2\inner{\va_k}{\vdelta}\bigr)$, whose curvature $\cos\!\bigl(2\inner{\va_k}{\vdelta}\bigr)$ keeps its sign only while $|\inner{\va_k}{\vdelta}|\le\pi/4$---that is, on $\Treg{1/2}$---and in return delivers an explicit rate.
Comparison-based selection instead sees only sublevel sets: the sandwich traps each basin between two balls, the radial monotonicity of $G$ on $\Treg{1}$ makes the basin a single star-shaped piece, and the absence of critical points inside $\Treg{1}$ keeps the success probability of small steps near $\tfrac12$. None of this requires convexity, so the black-box analysis runs on sublevel sets inside the quarter-wave region, at the price of a non-explicit rate.
The certified starting levels of the two algorithms are of the same order, $\lambda_r/128$. The periodicity of $G$, absent from the classical Sphere, affects only the black-box side: a mutation can jump to a lattice copy of the basin, whereas gradient descent never leaves its own. In both cases the relevant dimension is the rank $r$ rather than the nominal $N$.

Since these guarantees are uniform over their regions while perturbations---mutations, probes, noise---carry no structure, a natural stress test is a configuration extreme on as many shell terms as it can simultaneously saturate: the \emph{comb} $\vdelta=c\,((-1)^{i})_i$. Each shell form obeys $|\ell_k|\le2\norm{\vdelta}_\infty$, and the comb is precisely the sign pattern that saturates this bound on every term it activates ($\ell_k=\pm2c$ on the parity-mixed pairs of even shells, and $\ell_k=0$, hence inactive, elsewhere). Random perturbations do the opposite---they spread across the shell forms and populate the interior of the sandwich (cf.~Table~\ref{tab:sandwich-validation} below).
This alternating comb makes the tests concrete. Its directional curvature is $\tfrac{d^2}{dc^2}G(c\vu)=C\cos(4c)$ with $C=\vu^{\!\top}H\vu>0$, turning the Hessian indefinite the instant $c$ passes $\pi/8$---where $|\ell_k|=\pi/4$, the comb saturating the coordinate-to-argument bound. The three witness radii then stand in ratio $\pi/8:\pi/4:\pi/2 = 1:2:4$ for convexity, sandwich, and wrap (Figure~\ref{fig:comb}).

\section{Numerical Corroboration}\label{sec:numerical}
We corroborate the results at $N\in\{32,64,80,160\}$: the four-point shell representation, the sandwich, white-box gradient descent (Thms.~\ref{thm:whitebox-structure}--\ref{thm:whitebox-gd}), and black-box $(1{+}1)$-ES convergence (Theorem~\ref{thm:blackbox} and Proposition~\ref{prop:local}). Checks of identities that are proved exactly (T1--T3, S1--S3) are \emph{implementation} checks---agreement of independent evaluators---rather than evidence for the identities themselves.
All tests were implemented in \texttt{python3}.

\paragraph{Numerical validation of the shell structure.}
A five-part driver verifies the four-point shell representation and its corollaries; Table~\ref{tab:shell-validation} reports the worst case over four dimensions. 
This driver uses the circular (length-$N$ DFT) evaluation, a distinct variant in which the delay is not a symmetry: its Hessian has exactly one zero mode, whereas the canonical linear-convolution instance of Section~\ref{sec:def_qsphere} carries the full 2D gauge kernel (rank $N-2$; Table~\ref{tab:sandwich-validation}). 
The T2 ripple is carried exclusively by the wraparound quadruples of the circular (DFT) evaluation (four-point index sums differing by $\pm N$), and T4--T5 reproduce the plateau and smoothing curves published two decades ago by~\cite{ShirBaeck2007}.

\begin{table}[t]
\centering
\caption{Shell validation over $N\in\{32,64,80,160\}$.}
\label{tab:shell-validation}
\setlength{\tabcolsep}{4pt}
\begin{tabular}{@{}llr@{}}
\toprule
Test & Identity checked & Max.\ deviation \\
\midrule
T1 & shell sum $=$ DFT yield (rel.) & $3.3\times10^{-17}$ \\
T2 & exact gauges: const.\ / int.-slope (rel.\ gap) & $8.5\times10^{-16}$ \\
   & ripple ampl.\ vs.\ wraparound closed form & $10^{-13}$ (rel.) \\
T3 & $H = \sum_k 2w_k\,\va_k^{\vphantom{\top}}\va_k^{\!\top}$ vs.\ 2nd diff.\ (rel.) & $8.5\times10^{-8}$ \\
   & zero modes (gauge, circular eval.) & $1$ at every $N$ \\
   & $\kappa_{\mathrm{DW}}$ (nonzero modes) & $804$--$828$ \\
T4 & landscape's bottom $\E[Y]/Y(0)$, pairing vs.\ MC (2000 draws) & $8.6\times10^{-4}$ \\
T5 & Debye--Waller smoothing vs.\ MC, $\sigma\in[0,2]$ & $1.5\times10^{-2}$ \\
\bottomrule
\end{tabular}
\end{table}

\paragraph{Validating the sandwich.}
We validate the \textit{landscape structure} result---Cor.~\ref{cor:sandwich}, $(4/\pi^{2})\,Q \le G \le Q$ on $\Treg{1}=\{\min_{\vg\in\Gauge}\norm{\vdelta-\vg}_\infty\le\pi/4\}$---with a six-part driver run over uniform (flat), Gaussian, and random amplitude profiles (Table~\ref{tab:sandwich-validation}).
Two independent evaluators---the physical self-convolution yield and the shell $\sin^2$ representation---agree to machine precision, with exact gauge kernel and analytic Hessian of rank $N-2$ (S1--S3).
The sandwich itself shows \emph{zero} violations over $16{,}000$ uniform draws from $\Treg{1}$, with the observed ratio $G/Q$ concentrating well inside the certified interval $[4/\pi^{2},1]$ and narrowing as $N$ grows (S4). Both constants are sharp: the quarter-wave comb $(\pi/4)(-1)^{i}$ attains the lower constant $4/\pi^{2}$ \emph{exactly}, and $G/Q\to1$ as $\vdelta\to0$, with deficit $O(\norm{P\vdelta}^2)$ (S5). Beyond the region, the bound persists up to the pointwise Jordan condition $\max_k|\ell_k|\le\pi/2$ and fails past it: $\Treg{1}$ is the certified core, not an artifact of the proof (S6).

\begin{table}[t]
\centering
\caption{Sandwich-theorem validation: worst case over $N\in\{32,64,80,160\}$ and three amplitude profiles.}
\label{tab:sandwich-validation}
\setlength{\tabcolsep}{4pt}
\begin{tabular}{@{}llr@{}}
\toprule
 & Property checked & Result \\
\midrule
S1 & direct yield $=$ shell $\sin^2$ form (rel.) & $1.2\times10^{-15}$ \\
S2 & gauge kernel exact (incl.\ $\pi$-slope comb) & $\le 1.0\times10^{-9}$ \\
S3 & analytic vs.\ FD Hessian; rank $=N-2$ & $2.4\times10^{-6}$; exact \\
S4 & sandwich on $\Treg{1}$, $16{,}000$ draws & $0$ violations \\
   & \quad observed $G/Q$ (allowed $[0.405,1]$) & $[0.725,\,0.932]$ \\
S5 & lower const.\ attained by $\tfrac{\pi}{4}(-1)^i$ comb & $1.1\times10^{-16}$ \\
   & \quad upper const.\ approached as $\vdelta\to0$: $1-G/Q$ & $4.0\times10^{-9}$ \\
S6 & bound up to Jordan onset $\max|\ell|=\tfrac{\pi}{2}$ & confirmed \\
\bottomrule
\end{tabular}
\end{table}
\paragraph{White-box validation: gradient descent.}
The white-box driver works in the class form
$G(\vdelta)=\sum_k w_k\sin^2\!\langle \va_k,\vdelta\rangle$---a sum of nonnegative terms with no cancellation floor---allowing gradient descent to be tracked down to $10^{-24}\,G_0$
(Table~\ref{tab:whitebox-validation}).
For Theorem~\ref{thm:whitebox-structure}, the sampled Hessian is positive semidefinite throughout
$\Treg{1/2}$ and the convexity radius is \emph{exact}: along the alternating
comb $\vu=((-1)^i)$ the curvature is $(\vu^{\!\top}\!H\vu)\cos 4c$, vanishing
at $c=\pi/8$ (W1); the PL inequality
$\|\nabla G\|^2\ge(32\lambda_r/\pi^4)\,G$ shows zero violations over
$12{,}000$ draws, with the smallest sampled ratio ${\ge}6.4\times$ the
bound, and strong convexity modulo gauge holds at $\theta=1/4$ (W2).
The boundary is sharp in the strongest sense: the comb at $c=\pi/4$ is an exhibited critical point of the instance with $G>0$, so PL cannot extend to $\Treg{1}$ (W4). For Theorem~\ref{thm:whitebox-gd} under its verbatim hypotheses, the per-step contraction of $G$ respects $1-16/(\pi^4\kappa_r)$ at \emph{every} iteration, the invariant $\|P\vdelta^t\|_2\le\pi/16$ is never violated, and the gauge coordinates are frozen to machine precision---the white-box counterpart of the neutral gauge walk of the (1+1)-ES reported below (W3).

\begin{table}[t]
\centering
\caption{Validating Thms.~\ref{thm:whitebox-structure}--\ref{thm:whitebox-gd}. W1/W2/W4: $N\in\{32,64,80\}$;
W3: $N\in\{12,\dots,80\}$ ($\kappa_r=1.36$--$1.48$).}
\label{tab:whitebox-validation}
\setlength{\tabcolsep}{3pt}
\begin{tabular}{@{}llr@{}}
\toprule
 & Property checked & Result \\
\midrule
W1 & Hessian PSD on $\Treg{1/2}$ ($3{,}000$ pts/size) & confirmed \\
   & \quad comb curvature $=(\vu^{\!\top}\!H\vu)\cos 4c$ (rel.) & $2.3\times10^{-8}$ \\
W2 & PL on $\Treg{1/2}$, $12{,}000$ draws total & $0$ violations \\
   & \quad min ratio / bound $32\lambda_r/\pi^4$ & $\ge 6.4$ \\
   & \quad strong convexity margin at $\theta=\tfrac14$ & positive \\
W3 & per-step ratio $\le 1-16/(\pi^4\kappa_r)$ (to $10^{-24}G_0$) & $0$ violations \\
   & \quad invariant $\|P\vdelta^t\|_2\le\pi/16$; gauge drift & held; $\le2\!\times\!10^{-16}$ \\
W4 & comb at $c=\pi/4$: critical point, $G>0$ & $\|\nabla G\|\le5\!\times\!10^{-12}$ \\
   & \quad analytic PL-ratio overlay (rel.) & $3\times10^{-13}$ \\
\bottomrule
\end{tabular}
\end{table}

\paragraph{Black-box validation: (1+1)-ES.}
We run the algorithm of Theorem~\ref{thm:blackbox}---per-step success rule with $\aup=e^{0.1}$, $\adown=e^{-0.025}$ ($\ptgt=1/5$), the constants used by \citet{AAGM2022}---on the top-hat Quantum Sphere, initialized at a random direction on $\partial\Treg{1/2}$ ($\norm{P\vdelta^0}_\infty=\pi/8$, $G(\vdelta^0)/Y(0)\approx0.09$) with $\sigma_0=1.224\,\norm{P\vdelta^0}_2/r$, ten seeds per $N$. 
This start lies outside the certified sublevel set, since $G(\vdelta^0)\ge\frac{2\lambda_r}{\pi^2}\norm{P\vdelta^0}_2^2\ge\lambda_r/32>\bstar$ (Corollary~\ref{cor:sandwich}); by elitism and the Markov property, Theorem~\ref{thm:blackbox} applies from the first time $G(\vdelta^t)\le b$.
As Theorem~\ref{thm:blackbox} does not provide an explicit rate, we test what it and Proposition~\ref{prop:local} do assert.
\emph{(a)~Linear convergence on the certified set.} We fit the log-rate of $G_t$ on the window $10^{-9}Y(0)\le G\le b$, $b=\frac{\lambda_r}{128}\ptgt^{2/r}$, which lies inside the sublevel set of Theorem~\ref{thm:blackbox} (Table~\ref{tab:es-validation}). Every run is linear there (coefficient of determination $R^2\ge0.99$ for all $40$ fits), and the ES enters this set from an initial gap two to three orders of magnitude above it.
\emph{(b)~Terminal regime.} Running the same algorithm with the same seeds on the quadratic tangent $Q$ and on the isotropic Sphere gives rates within $2\%$ (on $Q$) and $4\%$ (on the Sphere) of those on $G$: the ES does not distinguish the Quantum Sphere from its tangent. 
Relative to the Sphere reference $2\psi^\star_{1/5}/r$ of Lemma~\ref{s:lem:rate}, the rate rises from $0.88$ ($r=30$) to $0.96$ ($r=158$); a rank sweep $r=\{8,\ldots,78\}$ shows the same rise, $0.60\to0.95$, identically for $G$ and $Q$.
\emph{(c)~Proposition~\ref{prop:local} directly.} Along the trajectories we draw extra candidates and count the steps on which $G$ and $Q$ disagree about acceptance (Table~\ref{tab:mismatch}). The disagreement vanishes as $G\to0$, roughly in proportion to $G$ (i.e., to $\norm{P\vdelta}^2$), with a single disagreement in $1.2\times10^6$ probes at $10^{-9}\le G/Y(0)<10^{-4}$; below $10^{-10}Y(0)$ the direct evaluation $Y(0)-Y$ loses precision to cancellation.
\emph{(d)~Step size and gauge.} In the fit window the success rate equals the balance identity $p=\ptgt-v/(2\ln(\aup/\adown))$, $v$ the measured log-rate of $G$, to within $5\times10^{-4}$ at every $N$, confirming that $\sigma_t$ tracks $\norm{P\vdelta^t}$ (normalized step $\sigma r/\rho=1.77$--$2.15$). The gauge coordinates move by $0.27$ during the first half of a run at $N=160$ and by $5\times10^{-4}$ during the second: a martingale that freezes as $\sigma_t\to0$ (Lemma~\ref{s:lem:proj}), in contrast with the frozen gauge of gradient descent (W3).
\emph{(e)~Traps.} Quasi-Newton (L-BFGS) runs from $2{,}150$ i.i.d.\ uniform random phase vectors---far-field starts with $G/Y(0)\approx0.95$---on top-hat instances with $N=8$--$160$ and on Gaussian and random amplitude profiles with $N\in\{32,64\}$ all reached the optimum lattice ($G<10^{-6}Y(0)$). We found no spurious local minimizer, consistent with $G_{\mathrm{trap}}=+\infty$; conditional on this, $\bstar=\frac{\lambda_r}{128}\ptgt^{2/r}$ and the barrier satisfies $\cstar\ge\lambda_r/8$ (Lemma~\ref{lem:barrier}).

\begin{table}[t]
\centering
\caption{Validating Theorem~\ref{thm:blackbox} and Proposition~\ref{prop:local} on the top-hat instance. Certified level $b=\frac{\lambda_r}{128}\ptgt^{2/r}$ and barrier bracket $\lambda_r/8\le\cstar\le\min_jG(\pi\ve_j)$ (lower end conditional on the absence of traps), relative to $Y(0)$; per-iteration log-rate of $G_t$ on the window $[10^{-9}Y(0),\,b]$, mean [min, max] over $10$ seeds ($\times10^{-3}$); ratio to the rate of the same ES on the quadratic tangent $Q$; and fraction of the Sphere reference $2\psi^\star_{1/5}/r$.}
\label{tab:es-validation}
\small
\setlength{\tabcolsep}{4pt}
\begin{tabular}{@{}rrcccccc@{}}
\toprule
$N$ & $r$ & $\kappa_r$ & $b/Y(0)$ & $\cstar/Y(0)$ & rate on $G$ & $G/Q$ & $v\,r/2\psi^\star_{1/5}$ \\
\midrule
$32$  & $30$  & $1.45$ & $6.8\times10^{-4}$ & $[0.012,\,0.17]$  & $11.06\;[9.72,\,12.03]$ & $1.00$ & $0.88$ \\
$64$  & $62$  & $1.48$ & $3.5\times10^{-4}$ & $[0.006,\,0.089]$ & $5.45\;[5.22,\,5.97]$  & $1.01$ & $0.90$ \\
$80$  & $78$  & $1.48$ & $2.9\times10^{-4}$ & $[0.005,\,0.072]$ & $4.45\;[4.01,\,4.70]$  & $1.00$ & $0.92$ \\
$160$ & $158$ & $1.49$ & $1.4\times10^{-4}$ & $[0.002,\,0.037]$ & $2.28\;[2.13,\,2.48]$  & $1.02$ & $0.96$ \\
\bottomrule
\end{tabular}
\end{table}

\begin{table}[t]
\centering
\caption{Proposition~\ref{prop:local} in action: frequency with which $G$ and $Q$ disagree on accepting the same candidate, along ES trajectories ($10$ seeds, $20$ probe candidates per iteration), binned by the gap level.}
\label{tab:mismatch}
\small
\setlength{\tabcolsep}{5pt}
\begin{tabular}{@{}lcccc@{}}
\toprule
$G/Y(0)$ & $[10^{-2},10^{-1})$ & $[10^{-3},10^{-2})$ & $[10^{-4},10^{-3})$ & $[10^{-9},10^{-4})$ \\
\midrule
$N=32$  & $1.1\times10^{-3}$ & $9.5\times10^{-5}$ & $2.3\times10^{-5}$ & $4.7\times10^{-6}$ ($1/2.1\times10^{5}$) \\
$N=160$ & $1.4\times10^{-3}$ & $1.4\times10^{-4}$ & $3.3\times10^{-5}$ & $0$ ($0/1.0\times10^{6}$) \\
\bottomrule
\end{tabular}
\end{table}


\section{Computational Challenges}\label{sec:challenges}
The results above concern the ideal, noiseless, dispersion-free landscape.
We close by naming what separates the theorems from the laboratory --- and, in one case, what separates the certified rates from the observed dynamics. 
We order the challenges by how much of the machinery already transfers: from a channel the present framework resolves exactly, through regimes it renders precise but leaves open, to benchmarks it makes constructible.

\paragraph{Additive input (actuator) noise: a closed channel.}
Let the applied phase be additively perturbed, $\phasev \mapsto
\phasev + \noisev$ with $\noisev \sim \mathcal{N}(\zerov, \Sigma)$
--- as in \cite{BOS04}. 
Since the yield is a trigonometric polynomial on the shell,
$Y(\phasev)=\tfrac12\sum_{k\in K} w_k\cos\bigl(2\inner{\va_k}{\phasev}\bigr)$
(the phase-blind pairs, for which $\va_k=\mathbf{0}$, carry the incoherent constant \( Y_{\mathrm{inc}}\)), 
the noise-averaged landscape is obtained \emph{exactly} by evaluating the noise's characteristic function at each shell vector: every weight acquires a Debye--Waller attenuation (with \( \tilde w_k = w_k\, e^{-2\va_k^{\!\top}\Sigma\,\va_k}\)),
\begin{align}
\label{eq:dw-reweight}
\E\bigl[Y(\phasev + \noisev)\bigr]
  &= \frac12\sum_{k \in K} \tilde w_k \cos\bigl(2\inner{\va_k}{\phasev}\bigr),
\end{align}
so the averaged landscape is the \emph{same} object with a reweighted shell: the maximizer is preserved, the flat directions are preserved, and every attained constant 
transfers under $w_k \mapsto \tilde w_k$.
The price is a residual yield floor $\tfrac12\sum_k (w_k - \tilde w_k)$, with leading behavior $\tfrac{1}{2}\tr(H\Sigma)$ --- the systematic loss term of the ESs' quality gain. 
In this input channel, noise deforms the \emph{landscape}, reweighting each shell term separately (rather than rescaling $Y$ globally), and does so in a non-disruptive manner. This describes the noise-averaged landscape; it does not by itself give guarantees for an optimizer that receives individual noisy evaluations. 

The channel that remains open is the output: \emph{signal} noise, where the observed yield is perturbed \cite{Arnold2002,BOS04}. 
We therefore pose an open question: by how much does the certified per-step contraction $1 - 16/(\pi^4\kappa_r)$ degrade as a function of the noise-to-signal ratio at the selection step, and does elitist success-based adaptation retain the terminal linear phase? Both the constants and the trust regions in which to pose this question are now explicit.

\paragraph{The far field: the bottom of the landscape.}
At the bottom of the landscape, climbing is not guaranteed anymore to be fast -- even when no traps are present, the gradient may be too small relative to the remaining yield, i.e., by a {\L}ojasiewicz-type inequality in the sense of \citep{lojasiewicz1963propriete} with no explicit exponent or constant available in this regime.
Our certified rates are valid at the top of the landscape; the empirical efficiency observed since the earliest SHG experiments~\cite{Shir08_QCE,Roslund2009b} extends far beyond it, with measured ascent trajectories staying nearly straight well outside any certified region~\cite{Nanduri2015}. 
Attaining a {\L}ojasiewicz inequality would bound the trajectory length those measurements find short relative to the displacement achieved --- and closing that gap, a quantitative plateau-exit estimate at the bottom of the landscape --- is, in our view, the right successor to the question of why effective controls appear easy to find~\cite{HoRabitz2006}.

\paragraph{Global certification: the Discrete Optimization route.}
Every certificate here is local, and the torus geometry says why it must be: $G$ depends on $\vdelta$ only through $e^{\mathrm{i}\vdelta}$, so optima form a lattice (per-pixel $2\pi$ wraps $\times$ the gauge plane), 
and past a computable wrap barrier---relative yield
$\varepsilon_{\mathrm{merge}}$, set by the weakest aperture-edge
amplitude, cf.\ the half-wrap bound of Lemma~\ref{lem:barrier}---basins merge and
no ball-shaped trust region survives in $\mathbb{R}^{N}$. Yet the same gap is exactly a weighted \emph{angular-synchronization} objective on the shell hypergraph~\cite{SINGER2011}, whose semidefinite lift over unit-diagonal PSD matrices is native to the torus, where the wrap lattice is invisible.
The open challenge is a tightness or no-spurious-minima guarantee for the shell-coupled instance beyond the trust region, attacking the far field from the algorithmic side: not how fast local search escapes the plateau, but whether the lift renders escape unnecessary.

\paragraph{SHG with Glass: a scalable dispersion benchmark.}
Next, we propose a concrete benchmark family. Inserting a dispersive slab (``glass'') of thickness $L$ into the beam path imprints a known quadratic spectral phase $\dispv(L)$ on the pulse. 
If the compensating phase is exactly representable by the shaper, the landscape is a pure translate, $Y_L(\phasev) = Y(\phasev + \dispv(L))$, with the global optimum shifted to the mask $\phasev^{\mathrm{tgt}} = -\dispv(L)$. Recovering it is a \emph{hidden-target} continuous problem in the full $N$-dimensional pixel space, a physically generated analogue of Boolean \textbf{OneMax} \cite{Droste2006} and \textbf{two-color Mastermind} \cite{doerr2013playingmastermindcolors}---hidden-target search whose black-box query complexity is classical---here with a target manufactured by a real optical element. 
Realization sharpens the challenge: a finite pixelized shaper with bounded per-pixel stroke cannot represent the continuous $\dispv$ exactly, so the reachable optimum is itself a constrained object. 
This yields a scalable family with independent, physically meaningful knobs---chirp magnitude ($L$), pixel count ($N$), and pixel-representation fidelity. Two properties make it a rigorous benchmark rather than a toy-model: the yield at the top $Y^{\star}$ is known in closed form; and the instance generator is fully specified by the physical forward model, ensuring reproducibility and dimensional scalability. The simulator infrastructure introduced herein supports the family directly, and we pose it as a challenge to the optimization community.

\paragraph{The minimization side: low-autocorrelation sequences.}
Writing $x_i=A_ie^{\mathrm{i}\delta_i}$ and $R_k=\sum_jx_{j+k}\cj{x_j}$ for the aperiodic autocorrelation, Parseval's identity gives
\[
Y(\vdelta)\;=\;\frac{1}{2\pi}\int_0^{2\pi}\Bigl|\sum_ix_ie^{\mathrm{i}ti}\Bigr|^4\dd t\;=\;R_0^2+2\sum_{k\ge1}|R_k|^2 ,
\qquad R_0=\sum_iA_i^2 ,
\]
so the phases act only through the sidelobes $R_{k\ge1}$. Maximizing $Y$---the Quantum Sphere---concentrates the autocorrelation, and its optimum is the transform-limited pulse. \emph{Minimizing} $Y$ flattens it: for $\vA=\mathbf 1$ and binary phases $\delta_i\in\{0,\pi\}$, $Y=N^2+2E$, where $E=\sum_{k\ge1}R_k^2$ is the energy of the low-autocorrelation binary sequence (LABS) problem \cite{Golay1977,LAC_Militzer}, a notoriously hard combinatorial benchmark \cite{PackebuschMertens2016} and a current testbed for quantum optimization \cite{Koch2026}; for continuous phases one obtains its polyphase counterpart. The same SHG apparatus therefore generates both an easy, analyzable maximization problem and, on binary phases, a notoriously hard minimization problem over the same shell representation (Theorem~\ref{thm:exact}), a physically realized control of difficulty that we pose as a further challenge.

\section{Conclusion}\label{sec:conclusion}
Two decades after the SHG landscape's trap-free topology was established, we supply its local metric geometry: an exact $\sin^2$ shell representation, a sharp two-sided quadratic sandwich, and linear convergence guarantees for both white-box and black-box optimization---with explicit constants on the white-box side, and on the black-box side a certified sublevel set whose level is set by $\lambda_r$ and, through the quarter-wave barrier, by the absence of traps. A single witness, the alternating comb, pins the convexity, sandwich and wrap radii in ratio $1{:}2{:}4$, and numerical simulations corroborate the predictions, including the terminal regime in which the Evolution Strategy sees only the quadratic tangent of the peak. Topology and metric thus enter together: trap-freeness is precisely what turns the sandwich into an explicit certified level for comparison-based search. The gap between topology and rate guarantees is closed with drift analysis for Evolution Strategies~\cite{MorinagaAkimoto2019,AAGM2022} and the Polyak--\L{}ojasiewicz framework~\citep{lojasiewicz1963propriete,karimi2016linear}---tools unavailable when the landscape was first analyzed.

\paragraph{Exact Gaussian smoothing.}
The statistical gradient of Roslund and Rabitz~\cite{Roslund2009}, demonstrated on this very SHG system, is the search gradient of ESs and the Gaussian-smoothing gradient of zeroth-order optimization~\cite{Wierstra2014,NesterovSpokoiny2017}: an unbiased estimate of $\nabla Y_{\sigma}$ for $Y_{\sigma}(\vdelta)=\E_{\vxi}[Y(\vdelta+\sigma\vxi)]$. What SHG adds is that this surrogate is known in closed form (Eq.~\ref{eq:dw-reweight}) --- so the laboratory estimator measures the \emph{exact} gradient of an explicitly known landscape.


\paragraph{Why SHG trajectories run straight.}
The sandwich suggests a mechanism for the empirically straight SHG ascent~\cite{Roslund2009b}, as quantified by the linearity ratio $R\ge1$ (path length over displacement) of \citet{Nanduri2015}. For the quadratic tangent, gradient flow is exactly $\vdelta(s)=e^{-Hs}\vdelta_0$, with $R=1$ if and only if $\vdelta_0$ is an eigenvector of $H$ modulo gauge; Cor.~\ref{cor:sandwich} keeps $G$ within a factor $\pi^2/4$ of this tangent throughout the quarter-wave region. This does not by itself establish the tangent flow for $G$, but it suggests that high-$R$ trajectories route near the \textit{torus saddles} formed by $2\pi$-wraps of weak edge pixels---the half-wraps of Lemma~\ref{lem:barrier}, and the same wrapping that laboratory ES campaigns had to address explicitly before ascent became routine \cite{Shir08_QCE}.

\paragraph{A synchronization instance for the AI community.} 
A profound insight emerges from the exact representation itself: the Fourier transform, which physically defines the yield, is not strictly necessary for the algorithmic analysis. Once the shell weights are isolated, the Quantum Sphere behaves at the \textit{top of the landscape} as a weighted angular-synchronization objective on the torus \cite{SINGER2011}. 
It thus aligns with the discrete-optimization family of Kuramoto \cite{Kuramoto1975,Strogatz2000} and Max-Cut \cite{Goemans1995} relaxations, distinguished primarily by the coupling of all shell angles through the linear images $\vv^{(s)} = D_s\vdelta$ of a single shared decision vector. 
The smallest nonzero Hessian eigenvalue $\lambda_r$ fulfills a dual role: it acts both as the PL modulus for the local convergence theory and as the Fiedler value \cite{Fiedler1973} governing relaxation tightness. This duality suggests that semidefinite lifting techniques \cite{Goemans1995,SINGER2011} could potentially extend our local trust-region certificates to global guarantees. 
Consequently, we propose the Quantum Sphere to the broad AI community in two distinct directions: (i) as an analyzable non-synthetic continuous benchmark; and (ii) as a synchronization discrete open problem. 

\section*{Acknowledgments}
The author is indebted to Jonathan Roslund, for his significant contributions that ignited this line of work.
During the preparation of this manuscript, the author utilized Google Gemini Pro and Anthropic Claude Opus to assist with readability and exposition. Following the use of these tools, the author thoroughly reviewed, edited, and verified all content, and takes full responsibility for the scientific results and final text.


\begin{thebibliography}{49}
\providecommand{\natexlab}[1]{#1}
\providecommand{\url}[1]{\texttt{#1}}
\expandafter\ifx\csname urlstyle\endcsname\relax
  \providecommand{\doi}[1]{doi: #1}\else
  \providecommand{\doi}{doi: \begingroup \urlstyle{rm}\Url}\fi

\bibitem[Shir and B\"{a}ck(2022)]{SeqExpEA_TutorialGECCO2022}
Ofer~M. Shir and Thomas B\"{a}ck.
\newblock Sequential experimentation by evolutionary algorithms.
\newblock In \emph{Proceedings of the Genetic and Evolutionary Computation
  Conference Companion}, GECCO '22, page 1450–1468, New York, NY, USA, 2022.
  Association for Computing Machinery.
\newblock ISBN 9781450392686.
\newblock \doi{10.1145/3520304.3533631}.
\newblock URL \url{https://doi.org/10.1145/3520304.3533631}.

\bibitem[{\L}ojasiewicz(1963)]{lojasiewicz1963propriete}
Stanis{\l}aw {\L}ojasiewicz.
\newblock Une propri{\'e}t{\'e} topologique des sous-ensembles analytiques
  r{\'e}els.
\newblock In \emph{Les {\'e}quations aux d{\'e}riv{\'e}es partielles}, volume
  117 of \emph{Colloques Internationaux du CNRS}, pages 87--89. {\'E}ditions du
  CNRS, Paris, 1963.

\bibitem[Polyak(1963)]{polyak1963gradient}
B.T. Polyak.
\newblock Gradient methods for the minimisation of functionals.
\newblock \emph{USSR Computational Mathematics and Mathematical Physics},
  3\penalty0 (4):\penalty0 864--878, 1963.
\newblock ISSN 0041-5553.
\newblock \doi{10.1016/0041-5553(63)90382-3}.

\bibitem[Karimi et~al.(2016)Karimi, Nutini, and Schmidt]{karimi2016linear}
Hamed Karimi, Julie Nutini, and Mark Schmidt.
\newblock Linear convergence of gradient and proximal-gradient methods under
  the polyak-\l{}ojasiewicz condition.
\newblock In \emph{European Conference on Machine Learning and Knowledge
  Discovery in Databases - Volume 9851}, ECML PKDD 2016, page 795–811,
  Berlin, Heidelberg, 2016. Springer-Verlag.
\newblock ISBN 9783319461274.
\newblock \doi{10.1007/978-3-319-46128-1_50}.
\newblock URL \url{https://doi.org/10.1007/978-3-319-46128-1_50}.

\bibitem[Boyd and Vandenberghe(2004)]{Boyd}
Stephen Boyd and Lieven Vandenberghe.
\newblock \emph{Convex Optimization}.
\newblock Cambridge University Press, New York, 2004.
\newblock ISBN 9781107394001.

\bibitem[Auger and Doerr(2011)]{Auger2011_TRS}
Anne Auger and Benjamin Doerr.
\newblock \emph{Theory of Randomized Search Heuristics: Foundations and Recent
  Developments}.
\newblock World Scientific Publishing Co., Inc., River Edge, NJ, USA, 2011.
\newblock ISBN 9789814282666.
\newblock \doi{10.1142/7438}.

\bibitem[B{\"a}ck et~al.(2013)B{\"a}ck, Foussette, and
  Krause]{Baeck2013contemporary}
Thomas B{\"a}ck, Christophe Foussette, and Peter Krause.
\newblock \emph{Contemporary Evolution Strategies}.
\newblock Natural Computing Series. Springer Berlin, Heidelberg, 2013.
\newblock ISBN 9783642401374.
\newblock \doi{10.1007/978-3-642-40137-4}.

\bibitem[Emmerich et~al.(2025)Emmerich, Shir, and Wang]{ESchapter2025}
Michael Emmerich, Ofer~M. Shir, and Hao Wang.
\newblock \emph{Evolution Strategies}, pages 89--123.
\newblock Springer Nature Switzerland, Cham, 2025.
\newblock ISBN 978-3-032-00385-0.
\newblock \doi{10.1007/978-3-032-00385-0_13}.

\bibitem[Rechenberg(1973)]{Rechenberg1973}
Ingo Rechenberg.
\newblock \emph{{Evolutionsstrategies: Optimierung technischer Systeme nach
  Prinzipien der biologischen Evolution}}.
\newblock Frommann-Holzboog Verlag, Stuttgart, Germany, 1973.

\bibitem[Beyer(2001)]{Beyer2001}
Hans-Georg Beyer.
\newblock \emph{The Theory of Evolution Strategies}.
\newblock Natural Computing Series. Springer, Heidelberg, 2001.
\newblock \doi{10.1007/978-3-662-04378-3}.

\bibitem[J{\"a}gersk{\"u}pper(2006)]{Jaegerskuepper2006}
Jens J{\"a}gersk{\"u}pper.
\newblock How the (1+1)~{ES} using isotropic mutations minimizes positive
  definite quadratic forms.
\newblock \emph{Theoretical Computer Science}, 361\penalty0 (1):\penalty0
  38--56, 2006.
\newblock \doi{10.1016/j.tcs.2006.04.004}.

\bibitem[Morinaga and Akimoto(2019)]{MorinagaAkimoto2019}
Daiki Morinaga and Youhei Akimoto.
\newblock Generalized drift analysis in continuous domain: linear convergence
  of (1+1)-es on strongly convex functions with lipschitz continuous gradients.
\newblock In \emph{Proceedings of the 15th ACM/SIGEVO Conference on Foundations
  of Genetic Algorithms}, FOGA '19, page 13–24, New York, NY, USA, 2019.
  Association for Computing Machinery.
\newblock ISBN 9781450362542.
\newblock \doi{10.1145/3299904.3340303}.
\newblock URL \url{https://doi.org/10.1145/3299904.3340303}.

\bibitem[Akimoto et~al.(2022)Akimoto, Auger, Glasmachers, and
  Morinaga]{AAGM2022}
Youhei Akimoto, Anne Auger, Tobias Glasmachers, and Daiki Morinaga.
\newblock Global linear convergence of evolution strategies on more than smooth
  strongly convex functions, 2022.
\newblock URL \url{https://arxiv.org/abs/2009.08647}.

\bibitem[Ollivier et~al.(2017)Ollivier, Arnold, Auger, and Hansen]{IGO}
Yann Ollivier, Ludovic Arnold, Anne Auger, and Nikolaus Hansen.
\newblock Information-geometric optimization algorithms: a unifying picture via
  invariance principles.
\newblock \emph{J. Mach. Learn. Res.}, 18\penalty0 (1):\penalty0 564–628,
  January 2017.
\newblock ISSN 1532-4435.
\newblock URL \url{http://jmlr.org/papers/v18/14-467.html}.

\bibitem[Koch et~al.(2022)Koch, Boscain, Calarco, et~al.]{Koch2022}
Christiane~P. Koch, Ugo Boscain, Tommaso Calarco, et~al.
\newblock {Quantum optimal control in quantum technologies. Strategic report on
  current status, visions and goals for research in Europe}.
\newblock \emph{EPJ Quantum Technol.}, 9\penalty0 (1):\penalty0 19, 2022.
\newblock \doi{10.1140/epjqt/s40507-022-00138-x}.

\bibitem[Rabitz et~al.(2004)Rabitz, Hsieh, and Rosenthal]{RHR2004}
Herschel Rabitz, Michael Hsieh, and Carey Rosenthal.
\newblock Quantum optimally controlled transition landscapes.
\newblock \emph{Science}, 303\penalty0 (5666):\penalty0 1998--2001, 2004.
\newblock \doi{10.1126/science.1093649}.

\bibitem[Pechen and Tannor(2011)]{PechenTannor2011}
Alexander~N. Pechen and David~J. Tannor.
\newblock Are there traps in quantum control landscapes?
\newblock \emph{Phys. Rev. Lett.}, 106:\penalty0 120402, Mar 2011.
\newblock \doi{10.1103/PhysRevLett.106.120402}.

\bibitem[Judson and Rabitz(1992)]{Hersch92}
Richard~S. Judson and Herschel Rabitz.
\newblock {Teaching Lasers to Control Molecules}.
\newblock \emph{Phys. Rev. Lett.}, 68\penalty0 (10):\penalty0 1500--1503, mar
  1992.
\newblock \doi{10.1103/PhysRevLett.68.1500}.

\bibitem[Roslund and Rabitz(2009{\natexlab{a}})]{Roslund2009}
Jonathan Roslund and Herschel Rabitz.
\newblock Gradient algorithm applied to laboratory quantum control.
\newblock \emph{Phys. Rev. A}, 79:\penalty0 053417, May 2009{\natexlab{a}}.
\newblock \doi{10.1103/PhysRevA.79.053417}.
\newblock URL \url{http://link.aps.org/doi/10.1103/PhysRevA.79.053417}.

\bibitem[Shir et~al.(2014)Shir, Roslund, Whitley, and Rabitz]{Shir-FOCAL}
Ofer~M. Shir, Jonathan Roslund, Darrell Whitley, and Herschel Rabitz.
\newblock {Efficient retrieval of landscape Hessian: Forced optimal covariance
  adaptive learning}.
\newblock \emph{Physical Review E}, 89:\penalty0 063306, Jun 2014.
\newblock \doi{10.1103/PhysRevE.89.063306}.
\newblock URL \url{http://link.aps.org/doi/10.1103/PhysRevE.89.063306}.

\bibitem[Ho and Rabitz(2006)]{HoRabitz2006}
Tak-San Ho and Herschel Rabitz.
\newblock {Why do Effective Quantum Controls Appear Easy to Find?}
\newblock \emph{Journal of Photochemistry and Photobiology A: Chemistry},
  180\penalty0 (3), Jun 2006.
\newblock \doi{10.1016/j.jphotochem.2006.03.038}.

\bibitem[Shir and B{\"a}ck(2007)]{ShirBaeck2007}
Ofer~M. Shir and Thomas B{\"a}ck.
\newblock The second harmonic generation case-study as a gateway for {ES} to
  quantum control problems.
\newblock In \emph{Proceedings of the 9th Annual Genetic and Evolutionary
  Computation Conference (GECCO 2007)}, pages 713--721, London, UK, 2007. ACM.
\newblock \doi{10.1145/1276958.1277104}.

\bibitem[Tannor and Rice(1985)]{Tannor85}
David~J. Tannor and Stuart~A. Rice.
\newblock {Control of Selectivity of Chemical Reaction via Control of Wave
  Packet Evolution}.
\newblock \emph{J. Chem. Phys.}, 83\penalty0 (10):\penalty0 5013--5018, 1985.
\newblock \doi{10.1063/1.449767}.

\bibitem[Rabitz et~al.(2000)Rabitz, de~Vivie-Riedle, Motzkus, and
  Kompa]{Hersch00}
Herschel Rabitz, Regina de~Vivie-Riedle, Marcus Motzkus, and Karl Kompa.
\newblock {Whither the Future of Controlling Quantum Phenomena?}
\newblock \emph{Science}, 288\penalty0 (5467):\penalty0 824--828, may 2000.
\newblock \doi{10.1126/science.288.5467.824}.

\bibitem[Nuernberger et~al.(2007)Nuernberger, Vogt, Brixner, and
  Gerber]{Gerber07}
Patrick Nuernberger, Gerhard Vogt, Tobias Brixner, and Gustav Gerber.
\newblock {Femtosecond Quantum Control of Molecular Dynamics in the Condensed
  Phase}.
\newblock \emph{Phys. Chem. Chem. Phys.}, 9\penalty0 (20):\penalty0 2470--2497,
  2007.
\newblock \doi{10.1039/B618760A}.

\bibitem[Roslund et~al.(2006)Roslund, Roth, and Rabitz]{Roslund06}
Jonathan Roslund, Matthias Roth, and Herschel Rabitz.
\newblock Laboratory observation of quantum control level sets.
\newblock \emph{Phys. Rev. A}, 74:\penalty0 043414, Oct 2006.
\newblock \doi{10.1103/PhysRevA.74.043414}.

\bibitem[Roslund et~al.(2009)Roslund, Shir, B{\"a}ck, and Rabitz]{LabES}
Jonathan Roslund, Ofer~M. Shir, Thomas B{\"a}ck, and Herschel Rabitz.
\newblock Accelerated optimization and automated discovery with covariance
  matrix adaptation for experimental quantum control.
\newblock \emph{Phys. Rev. A}, 80:\penalty0 043415, Oct 2009.
\newblock \doi{10.1103/PhysRevA.80.043415}.
\newblock URL \url{https://link.aps.org/doi/10.1103/PhysRevA.80.043415}.

\bibitem[Shir et~al.(2008)Shir, Roslund, B\"{a}ck, and Rabitz]{Shir08_QCE}
Ofer~M. Shir, Jonathan Roslund, Thomas B\"{a}ck, and Herschel Rabitz.
\newblock Performance analysis of derandomized evolution strategies in quantum
  control experiments.
\newblock In \emph{Proceedings of the 10th Annual Conference on Genetic and
  Evolutionary Computation}, GECCO '08, page 519–526, New York, NY, USA,
  2008. Association for Computing Machinery.
\newblock ISBN 9781605581309.
\newblock \doi{10.1145/1389095.1389193}.

\bibitem[Meshulach and Silberberg(1998)]{Meshulach1998}
Doron Meshulach and Yaron Silberberg.
\newblock Coherent quantum control of two-photon transitions by a femtosecond
  laser pulse.
\newblock \emph{Nature}, 396:\penalty0 239--242, 1998.
\newblock \doi{10.1038/24329}.

\bibitem[Singer(2011)]{SINGER2011}
Amit Singer.
\newblock Angular synchronization by eigenvectors and semidefinite programming.
\newblock \emph{Applied and Computational Harmonic Analysis}, 30\penalty0
  (1):\penalty0 20--36, 2011.
\newblock ISSN 1063-5203.
\newblock \doi{10.1016/j.acha.2010.02.001}.

\bibitem[Kern et~al.(2004)Kern, M{\"u}ller, Hansen, B{\"u}che, Ocenasek, and
  Koumoutsakos]{Kern2004}
Stefan Kern, Sibylle~D. M{\"u}ller, Nikolaus Hansen, Dirk B{\"u}che, Jiri
  Ocenasek, and Petros Koumoutsakos.
\newblock Learning probability distributions in continuous evolutionary
  algorithms -- a comparative review.
\newblock \emph{Natural Computing}, 3\penalty0 (1):\penalty0 77--112, 2004.
\newblock \doi{10.1023/B:NACO.0000023416.59689.4e}.

\bibitem[Morinaga et~al.(2021)Morinaga, Fukuchi, Sakuma, and
  Akimoto]{Morinaga2021}
Daiki Morinaga, Kazuto Fukuchi, Jun Sakuma, and Youhei Akimoto.
\newblock Convergence rate of the (1+1)-evolution strategy with success-based
  step-size adaptation on convex quadratic functions.
\newblock In \emph{Proceedings of the Genetic and Evolutionary Computation
  Conference}, GECCO '21, page 1169–1177, New York, NY, USA, 2021.
  Association for Computing Machinery.
\newblock ISBN 9781450383509.
\newblock \doi{10.1145/3449639.3459289}.
\newblock URL \url{https://doi.org/10.1145/3449639.3459289}.

\bibitem[Beyer et~al.(2004)Beyer, Olhofer, and Sendhoff]{BOS04}
Hans-Georg Beyer, Markus Olhofer, and Bernhard Sendhoff.
\newblock {On the Impact of Systematic Noise on the Evolutionary Optimization
  Performance---A Sphere Model Analysis}.
\newblock \emph{Genetic Programming and Evolvable Machines}, 5\penalty0
  (4):\penalty0 327--360, 2004.
\newblock ISSN 1389-2576.
\newblock \doi{10.1023/B:GENP.0000036020.79188.a0}.

\bibitem[Arnold(2002)]{Arnold2002}
Dirk~V. Arnold.
\newblock \emph{Noisy Optimization with Evolution Strategies}.
\newblock Genetic Algorithms and Evolutionary Computation. Springer, Berlin,
  2002.
\newblock ISBN 978-1-4020-7105-8.

\bibitem[Roslund and Rabitz(2009{\natexlab{b}})]{Roslund2009b}
Jonathan Roslund and Herschel Rabitz.
\newblock Experimental quantum control landscapes: Inherent monotonicity and
  artificial structure.
\newblock \emph{Phys. Rev. A}, 80:\penalty0 013408, Jul 2009{\natexlab{b}}.
\newblock \doi{10.1103/PhysRevA.80.013408}.
\newblock URL \url{http://link.aps.org/doi/10.1103/PhysRevA.80.013408}.

\bibitem[Nanduri et~al.(2015)Nanduri, Shir, Donovan, Ho, and
  Rabitz]{Nanduri2015}
Arun Nanduri, Ofer~M. Shir, Ashley Donovan, Tak-San Ho, and Herschel Rabitz.
\newblock Exploring the complexity of quantum control optimization
  trajectories.
\newblock \emph{Phys. Chem. Chem. Phys.}, 17:\penalty0 334--347, 2015.
\newblock \doi{10.1039/C4CP03853C}.

\bibitem[Droste et~al.(2006)Droste, Jansen, and Wegener]{Droste2006}
Stefan Droste, Thomas Jansen, and Ingo Wegener.
\newblock Upper and lower bounds for randomized search heuristics in black-box
  optimization.
\newblock \emph{Theory of Computing Systems}, 39\penalty0 (4):\penalty0
  525--544, 2006.
\newblock \doi{10.1007/s00224-004-1177-z}.

\bibitem[Doerr et~al.(2013)Doerr, Doerr, Spöhel, and
  Thomas]{doerr2013playingmastermindcolors}
Benjamin Doerr, Carola Doerr, Reto Spöhel, and Henning Thomas.
\newblock Playing mastermind with many colors, 2013.
\newblock URL \url{https://arxiv.org/abs/1207.0773}.

\bibitem[Golay(1977)]{Golay1977}
M.~Golay.
\newblock Sieves for low autocorrelation binary sequences.
\newblock \emph{IEEE Transactions on Information Theory}, 23\penalty0
  (1):\penalty0 43--51, 1977.
\newblock \doi{10.1109/TIT.1977.1055653}.

\bibitem[Militzer et~al.(1998)Militzer, Zamparelli, and Beule]{LAC_Militzer}
Burkhard Militzer, Michele Zamparelli, and Dieter Beule.
\newblock Evolutionary search for low autocorrelated binary sequences.
\newblock \emph{IEEE Transactions on Evolutionary Computation}, 2\penalty0
  (1):\penalty0 34--39, 1998.
\newblock ISSN 1089-778X.
\newblock \doi{10.1109/4235.728212}.

\bibitem[Packebusch and Mertens(2016)]{PackebuschMertens2016}
Tom Packebusch and Stephan Mertens.
\newblock Low autocorrelation binary sequences.
\newblock \emph{Journal of Physics A: Mathematical and Theoretical},
  49\penalty0 (16):\penalty0 165001, 2016.

\bibitem[Koch et~al.(2026)Koch, Bernal~Neira, Chen, Cortiana, Egger, Heese,
  Hegade, Gomez~Cadavid, Huang, Itoko, Kleinert, Maciel~Xavier, Mohseni,
  Montanez-Barrera, Nakano, Nannicini, O'Meara, Pauckert, Proissl, Ramesh,
  Schicker, Shimada, Takeori, Valls, Van~Bulck, Woerner, and Zoufal]{Koch2026}
Thorsten Koch, David~E. Bernal~Neira, Ying Chen, Giorgio Cortiana, Daniel~J.
  Egger, Raoul Heese, Narendra~N. Hegade, Alejandro Gomez~Cadavid, Rhea Huang,
  Toshinari Itoko, Thomas Kleinert, Pedro Maciel~Xavier, Naeimeh Mohseni,
  Jhon~A. Montanez-Barrera, Koji Nakano, Giacomo Nannicini, Corey O'Meara,
  Justin Pauckert, Manuel Proissl, Anurag Ramesh, Maximilian Schicker, Noriaki
  Shimada, Mitsuharu Takeori, V{\'i}ctor Valls, David Van~Bulck, Stefan
  Woerner, and Christa Zoufal.
\newblock The quantum optimization benchmarking library.
\newblock \emph{Nature Computational Science}, 6\penalty0 (6):\penalty0
  653--671, June 2026.
\newblock ISSN 2662-8457.
\newblock \doi{10.1038/s43588-026-00991-1}.
\newblock URL \url{https://doi.org/10.1038/s43588-026-00991-1}.

\bibitem[Wierstra et~al.(2014)Wierstra, Schaul, Glasmachers, Sun, Peters, and
  Schmidhuber]{Wierstra2014}
Daan Wierstra, Tom Schaul, Tobias Glasmachers, Yi~Sun, Jan Peters, and
  J\"{u}rgen Schmidhuber.
\newblock Natural evolution strategies.
\newblock \emph{J. Mach. Learn. Res.}, 15\penalty0 (1):\penalty0 949–980,
  January 2014.
\newblock ISSN 1532-4435.
\newblock URL \url{http://jmlr.org/papers/v15/wierstra14a.html}.

\bibitem[Nesterov and Spokoiny(2017)]{NesterovSpokoiny2017}
Yurii Nesterov and Vladimir Spokoiny.
\newblock Random gradient-free minimization of convex functions.
\newblock \emph{Found. Comput. Math.}, 17\penalty0 (2):\penalty0 527–566,
  April 2017.
\newblock ISSN 1615-3375.
\newblock \doi{10.1007/s10208-015-9296-2}.
\newblock URL \url{https://doi.org/10.1007/s10208-015-9296-2}.

\bibitem[Kuramoto(1975)]{Kuramoto1975}
Yoshiki Kuramoto.
\newblock Self-entrainment of a population of coupled non-linear oscillators.
\newblock In Huzihiro Araki, editor, \emph{International Symposium on
  Mathematical Problems in Theoretical Physics}, pages 420--422, Berlin,
  Heidelberg, 1975. Springer Berlin Heidelberg.
\newblock ISBN 978-3-540-37509-8.

\bibitem[Strogatz(2000)]{Strogatz2000}
Steven~H. Strogatz.
\newblock From kuramoto to crawford: exploring the onset of synchronization in
  populations of coupled oscillators.
\newblock \emph{Physica D: Nonlinear Phenomena}, 143\penalty0 (1):\penalty0
  1--20, 2000.
\newblock ISSN 0167-2789.
\newblock \doi{https://doi.org/10.1016/S0167-2789(00)00094-4}.

\bibitem[Goemans and Williamson(1995)]{Goemans1995}
Michel~X. Goemans and David~P. Williamson.
\newblock Improved approximation algorithms for maximum cut and satisfiability
  problems using semidefinite programming.
\newblock \emph{J. ACM}, 42\penalty0 (6):\penalty0 1115–1145, November 1995.
\newblock ISSN 0004-5411.
\newblock \doi{10.1145/227683.227684}.
\newblock URL \url{https://doi.org/10.1145/227683.227684}.

\bibitem[Fiedler(1973)]{Fiedler1973}
Miroslav Fiedler.
\newblock {Algebraic Connectivity of Graphs}.
\newblock \emph{Czechoslovak Mathematical Journal}, 23\penalty0 (98):\penalty0
  298--305, 1973.

\bibitem[Anderson(1955)]{Anderson1955}
T.~W. Anderson.
\newblock The integral of a symmetric unimodal function over a symmetric convex
  set and some probability inequalities.
\newblock \emph{Proceedings of the American Mathematical Society}, 6\penalty0
  (2):\penalty0 170--176, 1955.
\newblock ISSN 00029939, 10886826.
\newblock URL \url{http://www.jstor.org/stable/2032333}.

\end{thebibliography}

\clearpage
\appendix
\numberwithin{equation}{section}
\numberwithin{lemma}{section}
\numberwithin{proposition}{section}
\numberwithin{corollary}{section}
\numberwithin{remark}{section}
\numberwithin{assumption}{section}
\numberwithin{example}{section}

\section{Setting and notation}\label{s:setting}

The Quantum Sphere's yield gap belongs to the weighted $\sin^2$ class,
\begin{equation}
G(\vdelta)\;=\;\sum_{k\in K} w_k\,\sin^{2}\!\bigl(\inner{\va_k}{\vdelta}\bigr),
\qquad w_k\ge0,
\label{s:class}
\end{equation}
with $k=(s,i,p)$ ranging over ordered pairs of the shells $S_s=\{i:\,0\le i\le N{-}1,\ 0\le s{-}i\le N{-}1\}$, and given $\ve_i$ as the $i$-th standard basis vector of $\R^N$:
\[
\va_k=\tfrac12\bigl(\ve_i+\ve_{s-i}-\ve_p-\ve_{s-p}\bigr),
\qquad
w_k = 2\,A_iA_{s-i}A_pA_{s-p},
\]
where $Y(\vdelta)=\sum_s\bigl|\sum_{i\in S_s}A_iA_{s-i}
e^{\mathrm{i}(\delta_i+\delta_{s-i})}\bigr|^2$ and $G=Y(0)-Y$; the
explicit expansion of \eqref{s:class} is Corollary~\ref{cor:sin2-disc}.
Shell
phases and weights: $v^{(s)}_i=\delta_i+\delta_{s-i}$,
$\alpha^{(s)}_i=A_iA_{s-i}$; note $|\inner{\va_k}{\vdelta}|\le
2\norm{\vdelta}_\infty$, write $\ell_k:=\inner{\va_k}{\vdelta}$ for the
shell linear forms, and define the \emph{incoherent yield}
\begin{equation}
Y_{\mathrm{inc}}
:=2\sum_{s}\sum_{i\in S_s}\bigl(\alpha^{(s)}_i\bigr)^{2}
-\sum_{s\ \mathrm{even}}\bigl(\alpha^{(s)}_{s/2}\bigr)^{2}.
\label{s:yinc}
\end{equation}
The quadratic tangent and Hessian at the optimum are
\begin{equation}
Q(\vdelta)=\sum_k w_k\inner{\va_k}{\vdelta}^{2}
=\tfrac12\vdelta^{\!\top}\!H\vdelta,
\qquad
H=\sum_k 2w_k\,\va_k^{\vphantom{\top}}\va_k^{\!\top},
\label{s:QH}
\end{equation}
with rank $r$, extreme nonzero eigenvalues $\lambda_1\ge\lambda_r>0$,
$\kappa_r=\lambda_1/\lambda_r$, null space $\Gauge:=\ker H$, and $\Proj$
the orthogonal projector onto $\Gauge^\perp$. Restricted to the support
window $W$, $\Gauge$ is exactly the two-dimensional space of affine
phases (Corollary~\ref{cor:gauge}); coordinates outside $W$ carry zero amplitude and
are inert, so $\dim\Gauge=N-|W|+2$ in general, and $\Gauge$ is the affine
gauge proper when $W=\{0,\dots,N-1\}$. \\
For symmetric matrices, $A\preceq B$ will denote the Loewner (positive-semidefinite) order, i.e.\ the difference $B-A\succeq0$ is positive-semidefinite. 

\medskip

\noindent We define the \textit{region family}:
\begin{equation}
\Treg{\theta}=\Bigl\{\vdelta:\ \min_{\vg\in\Gauge}
\norm{\vdelta-\vg}_\infty\le\theta\pi/4\Bigr\},\qquad\theta\in(0,1],
\label{s:regions}
\end{equation}
as in the main text. $\vdelta\in\Treg{\theta}$ implies
$|\inner{\va_k}{\vdelta}|\le\theta\pi/2$ for all $k$ ($\va_k\perp\Gauge$);
and since $\min_{\vg\in\Gauge}\norm{\vdelta-\vg}_\infty\le
\norm{\Proj\vdelta}_\infty\le\norm{\Proj\vdelta}_2$, an $\ell_2$ bound
on $\Proj\vdelta$ certifies membership---the form used in the Theorem~\ref{thm:whitebox-gd} induction.

\medskip

\noindent Throughout, the support of $\vA$ is a contiguous index window $W$ with
$|W|\ge3$, as in the main manuscript.


\newpage

\section{Proofs of the corollaries of Theorem~\ref{thm:exact}}
\label{s:thm12}

Theorem~\ref{thm:exact} is stated and proved in the main text; here we prove its
three corollaries.

\begin{restated}{Corollary~\ref{cor:sin2-disc} (restated)}
For every $\vdelta\in\R^{N}$,
\begin{equation}
\label{s:sin2-sum}
G(\vdelta) = 2\sum_{s=0}^{2N-2} \sum_{i,p\,\in S_s} A_i A_{s-i}\,A_p A_{s-p} \cdot \sin^{2}\!\Bigl( \tfrac{\delta_i+\delta_{s-i}-\delta_p-\delta_{s-p}}{2} \Bigr),
\end{equation}
where the inner sum runs over all ordered pairs $(i,p) \in S_s \times S_s$.
\end{restated}

\begin{proof}
Define the per-shell phases $v^{(s)}_i=\delta_i+\delta_{s-i}$ and
weights $\alpha^{(s)}_i=A_iA_{s-i}$ for $i\in S_s$. The gap partitions
as $G=\sum_s G_s$, where
\begin{equation}
\begin{aligned}
G_s(\vdelta) &= \Bigl(\sum_{i\in S_s}\alpha^{(s)}_i\Bigr)^{\!2}
- \Bigl|\sum_{i\in S_s}\alpha^{(s)}_i e^{\,\mathrm{i}v^{(s)}_i}\Bigr|^{2} \\
&= 2\sum_{i,p\in S_s} \alpha^{(s)}_i\alpha^{(s)}_p
\sin^{2}\!\Bigl(\tfrac{v^{(s)}_i-v^{(s)}_p}{2}\Bigr).
\end{aligned}
\end{equation}
Summing over $s$ yields \eqref{s:sin2-sum}.
\end{proof}

\begin{restated}{Corollary~\ref{cor:gauge} (restated)}
Suppose the support of $\vA$ is a contiguous index window $W$ with
$|W|\ge3$. Then $G(\vdelta)=0$ if and only if $\vdelta$ is affine
\emph{modulo $2\pi$} on $W$: there exist $a,b\in\R$ with
$\delta_i\equiv a+b\,i \pmod{2\pi}$ for all $i\in W$. Locally, the
kernel of $H=\nabla^2G(0)$ restricted to $W$ is exactly the
two-dimensional space of affine phases.
\end{restated}

\begin{proof}
To obtain $G(\vdelta)=0$ we must show that every term in \eqref{s:sin2-sum} having a positive weight vanishes, i.e.\ $v^{(s)}_i=\delta_i+\delta_{s-i}$ is constant
modulo $2\pi$ in $i$ on each shell $S_s\cap W$. Fix an interior index
$i$ of $W$---one exists since $|W|\ge3$---and take the even shell
$s=2i$, which contains the pairs $(i,i)$ and $(i-1,i+1)$ with positive
weights $A_i^2$ and $A_{i-1}A_{i+1}$. 
Equality of their $v^{(s)}$-values gives
$2\delta_i \equiv \delta_{i-1}+\delta_{i+1} \pmod{2\pi}$, i.e.\ the
discrete second difference of $\vdelta$ vanishes,
\begin{equation*}
\delta_{i-1}-2\delta_i+\delta_{i+1} \;\equiv\; 0 \pmod{2\pi}.
\end{equation*}
Holding at every interior index of $W$, this is exactly affinity:
$\delta_i\equiv a+b\,i\pmod{2\pi}$ for some $a,b\in\R$. The converse is
the invariance of $Y$ under global phase, delay, and per-pixel $2\pi$
shifts ($Y$ depends on $\vdelta$ only through $e^{\mathrm{i}\vdelta}$).
In the continuum: for measurable $\delta$, $G=0$ iff
$e^{i\delta(\omega)}=e^{i(a+b\omega)}$ a.e.\ (measurable solutions of
the circle-valued Cauchy equation).

\noindent For the kernel, the analytic SHG Hessian was already derived in closed form in \citep{Shir-FOCAL}. 
That Hessian was computed at $\vdelta^\star=0$, but since $G=Y(0)-Y$, the calculation herein coincides with it up to the sign of the objective: $H=\nabla^2 G(0) =-H_{\mathrm{SHG}}$.\footnote{Expanding $\tfrac12\int\dd\omega\iint
\mu_\omega(\omega')\mu_\omega(\omega'')\,\Delta_\omega^2$ about
$\vdelta=0$ reproduces the diagonal and two off-diagonal terms of the analytic SHG's Hessian with the opposite sign; see Remark~\ref{rem:focal-check}. }
Concretely, $\ker H=\bigcap_{w_k>0}\ker\inner{\va_k}{\cdot}$ by \eqref{s:QH}, and the same second-difference argument---now over $\R$, since $\inner{\va_k}{\vdelta}=0$ exactly---leaves precisely the affine phases.
\end{proof}

\begin{restated}{Corollary~\ref{cor:dom} (restated)}
With $Y_{\mathrm{inc}}$ as in \eqref{s:yinc}, for every $\vdelta\in\R^{N}$ we have
$G(\vdelta)\le Q(\vdelta)$ and $G(\vdelta)\le 2\bigl(Y(0)-Y_{\mathrm{inc}}\bigr)$;
moreover, when the coordinates of $\vdelta$ are i.i.d.\ uniform on $[-\pi,\pi]$,
$\E[G(\vdelta)]=Y(0)-Y_{\mathrm{inc}}$ exactly.
\end{restated}

\begin{proof}
Write the pair index $k=(s,i,p)$ with weight
$w_k=2\alpha^{(s)}_i\alpha^{(s)}_p\ge0$ and half-difference
$\ell_k=\tfrac12\bigl(v^{(s)}_i-v^{(s)}_p\bigr)$, so that
\eqref{s:sin2-sum} reads $G=\sum_{k\in K}w_k\sin^2\ell_k$ and
$Q(\vdelta)=\sum_{k\in K}w_k\ell_k^2$, consistently with \eqref{s:QH}.

\emph{Domination.} Since $\sin^2\ell_k\le\ell_k^2$ and $w_k\ge0$,
summing over $K$ gives $G\le Q$. The deficit
$Q-G=\sum_k w_k\bigl(\ell_k^2-\sin^2\ell_k\bigr)$ is a sum of
nonnegative terms; as $x^2-\sin^2x=\tfrac13x^4+O(x^6)$, it is quartic to
leading order, so $Q$ and $G$ agree through third order and $Q$ is the
genuine $2$-jet of $G$ at the optimum.

\emph{Responsive weight.} Sort the terms of \eqref{s:sin2-sum} by phase
dependence. Each shell $S_s$ is symmetric, so the reversed index
$p=s-i$ gives the same shell phase,
$v^{(s)}_{s-i}=\delta_{s-i}+\delta_i=v^{(s)}_i$; hence pairs with
$p\in\{i,s-i\}$ contribute nothing to $G$ for any $\vdelta$. Using
$\alpha^{(s)}_{s-i}=\alpha^{(s)}_i$, the total weight of these
phase-blind pairs is exactly $2Y_{\mathrm{inc}}$ with $Y_{\mathrm{inc}}$
as in \eqref{s:yinc} (the center $i=s/2$ of an even shell is its own
reverse, counted once); equivalently $Y_{\mathrm{inc}}=\E[Y]$ under
uniform phases. The responsive remainder therefore carries total weight
$2\bigl(Y(0)-Y_{\mathrm{inc}}\bigr)$, and bounding $\sin^2(\cdot)\le1$ on
it yields $G\le2(Y(0)-Y_{\mathrm{inc}})$. 
The responsive weight is nonetheless the correct normalizer for the expectation, which is what we use below.

\emph{Expectation.} Let $\vdelta\sim\mathcal{U}[-\pi,\pi]^N$ i.i.d.\ and
take a responsive term, $\{p,s-p\}\neq\{i,s-i\}$.\\ 
Its argument $2\ell_k=v^{(s)}_i-v^{(s)}_p$ is a nonzero integer combination of
independent uniform phases, so $\E[\cos 2\ell_k]=\prod_j\frac{\sin(c_j\pi)}{c_j\pi}=0$ (each nonzero integer coefficient $c_j$ kills its factor over the full $2\pi$
period), whence $\E[\sin^2\ell_k]=\tfrac12(1-\E[\cos2\ell_k])=\tfrac12$. Summing over
responsive pairs, $\E[G]=2\cdot\tfrac12\cdot\bigl(Y(0)-Y_{\mathrm{inc}}\bigr) =Y(0)-Y_{\mathrm{inc}}$---exactly half the responsive weight.
\end{proof}

\begin{restated}{Corollary~\ref{cor:sandwich} (restated)}
On the quarter-wave region $\Treg{1}$ of \eqref{s:regions},
$\tfrac{4}{\pi^{2}}\,Q(\vdelta)\le G(\vdelta)\le Q(\vdelta)$, the upper
bound holding globally.
\end{restated}

\begin{proof}
On $\Treg{1}$ every argument satisfies $|\inner{\va_k}{\vdelta}|\le\pi/2$
(the $\va_k$ annihilate $\Gauge$, so the gauge minimization in
\eqref{s:regions} is free), where Jordan's inequality
$\sin^2x\ge(4/\pi^2)x^2$ applies termwise; summing against $w_k\ge0$
gives the lower bound. The upper bound is Corollary~\ref{cor:dom}. In particular
the quadratic-growth bound
$G(\vdelta)\ge\tfrac{2\lambda_r}{\pi^2}\norm{\Proj\vdelta}_2^2$ holds on
$\Treg{1}$.
\end{proof}

\section{Auxiliary lemmas}\label{s:aux}

\begin{lemma}[Derivatives; gauge frozen]\label{s:lem:deriv}
\[
\nabla G(\vdelta)=\sum_k w_k\sin\bigl(2\inner{\va_k}{\vdelta}\bigr)\,\va_k,
\qquad
\nabla^2G(\vdelta)=\sum_k 2w_k\cos\bigl(2\inner{\va_k}{\vdelta}\bigr)\,
\va_k^{\vphantom{\top}}\va_k^{\!\top}.
\]
In particular
$\nabla G(\vdelta)\in\operatorname{range}H=\Gauge^\perp$ for every
$\vdelta$: gradient trajectories never move the gauge coordinates.
\end{lemma}
\begin{proof}
Differentiate \eqref{s:class} termwise;
$\tfrac{d}{dx}\sin^2x=\sin2x$ and $\tfrac{d^2}{dx^2}\sin^2x=2\cos2x$.
Each term of $\nabla G$ is a multiple of $\va_k\in\Gauge^\perp$.
\end{proof}

\begin{lemma}[Global smoothness]\label{s:lem:smooth}
$-H\preceq\nabla^2G(\vdelta)\preceq H$ for all $\vdelta$; hence
$\nabla G$ is $\lambda_1$-Lipschitz and
$\norm{\nabla G(\vdelta)}_2\le\lambda_1\norm{\Proj\vdelta}_2$.
\end{lemma}
\begin{proof}
By Lemma~\ref{s:lem:deriv} and $|\cos|\le1$,
$\pm\nabla^2G(\vdelta)\preceq\sum_k2w_k\va_k^{\vphantom{\top}}\va_k^{\!\top}=H$
in the positive-semidefinite order. For the last claim write
$\nabla G(\vdelta)=\nabla G(\Proj\vdelta)-\nabla G(0)$ (gauge invariance,
$\nabla G(0)=0$) and apply the Lipschitz bound.
\end{proof}

\begin{lemma}[Star-convexity toward the optimum]\label{s:lem:star}
For $\theta\in(0,1)$ and $\vdelta\in\Treg{\theta}$,
\begin{equation}
\inner{\nabla G(\vdelta)}{\vdelta}\;\ge\;
\frac{2\sin(\theta\pi)}{\theta\pi}\,Q(\vdelta)\;\ge\;
\frac{2\sin(\theta\pi)}{\theta\pi}\,G(\vdelta)
\;\;\stackrel{\theta=1/2}{=}\;\;\frac{4}{\pi}\,G(\vdelta).
\label{s:star}
\end{equation}
\end{lemma}
\begin{proof}
On $\Treg{\theta}$, $|2\inner{\va_k}{\vdelta}|\le\theta\pi\le\pi$, and
$y\mapsto\sin(y)/y$ is decreasing on $[0,\pi]$, so termwise
$\sin(2\ell_k)\,\ell_k=2\ell_k^2\cdot\tfrac{\sin(2\ell_k)}{2\ell_k}
\ge\tfrac{2\sin(\theta\pi)}{\theta\pi}\ell_k^2$. Summing gives the
first inequality; $Q\ge G$ is Corollary~\ref{cor:dom}.
\end{proof}

\section{Proof of Theorem~\ref{thm:whitebox-structure} (white-box: convexity and PL)}\label{s:pl}

\begin{restated}{Theorem~\ref{thm:whitebox-structure} (restated)}
On $\Treg{1/2}$ the gap is convex and satisfies
$\norm{\nabla G}_2^2\ge\tfrac{32\lambda_r}{\pi^4}G$; on
$\Treg{\theta}$, $\theta<\tfrac12$, it is
$\cos(\theta\pi)\lambda_r$-strongly convex modulo gauge. 
Both radii are sharp already for the Quantum Sphere instance, hence a
fortiori for the class \eqref{s:class}. The alternating comb at the
quarter-wave radius is a critical point with positive gap on
$\partial\Treg{1}$, so no star-convexity or PL constant survives at
$\theta=1$ (Proposition~\ref{cx:crit} below; cf.\ the graded constant of
Corollary~\ref{s:cor:gradedpl}, which degrades to $0$ as $\theta\to1$).
The same comb renders the Hessian indefinite on every $\Treg{\theta}$ with
$\theta>\tfrac12$, so the convexity radius $\tfrac12$ cannot be enlarged
(Proposition~\ref{cx:prop}).
\end{restated}

\begin{proof}[Proof of convexity and strong convexity]
On $\Treg{\theta}$ with $\theta\le\tfrac12$,
$|2\inner{\va_k}{\vdelta}|\le\theta\pi\le\pi/2$, so
$\cos(2\inner{\va_k}{\vdelta})\ge\cos(\theta\pi)\ge0$ termwise; by
Lemma~\ref{s:lem:deriv},
$\nabla^2G(\vdelta)\succeq\cos(\theta\pi)H\succeq0$. For
$\theta<\tfrac12$ this gives modulus $\cos(\theta\pi)\lambda_r$ on
$\Gauge^\perp$.
\end{proof}

\begin{proof}[Proof of the PL inequality]
If $G(\vdelta)=0$ there is nothing to prove; otherwise
$\Proj\vdelta\neq0$. By Cauchy--Schwarz, Lemma~\ref{s:lem:deriv}
($\nabla G\perp\Gauge$), and Lemma~\ref{s:lem:star} at
$\theta=\tfrac12$,
\[
\norm{\nabla G(\vdelta)}_2\;\ge\;
\frac{\inner{\nabla G(\vdelta)}{\vdelta}}{\norm{\Proj\vdelta}_2}
\;\ge\;\frac{4}{\pi}\,\frac{G(\vdelta)}{\norm{\Proj\vdelta}_2}\,.
\]
By quadratic growth (Corollary~\ref{cor:sandwich}, valid since
$\Treg{1/2}\subset\Treg{1}$),
$\norm{\Proj\vdelta}_2\le\tfrac{\pi}{\sqrt{2\lambda_r}}\sqrt{G(\vdelta)}$;
combining and squaring gives
$\norm{\nabla G}_2^2\ge\tfrac{32\lambda_r}{\pi^4}G$.
\end{proof}

\begin{corollary}[Graded PL up to the quarter-wave boundary]
\label{s:cor:gradedpl}
For every $\theta\in(0,1)$ and $\vdelta\in\Treg{\theta}$,
\[
\norm{\nabla G(\vdelta)}_2^{\,2}\;\ge\;
\frac{8\sin^2(\theta\pi)}{\theta^2\pi^4}\,\lambda_r\,G(\vdelta),
\]
recovering the constant $32\lambda_r/\pi^4$ at $\theta=\tfrac12$ and
degrading continuously to $0$ as $\theta\to1$, in agreement with the critical point of Proposition~\ref{cx:crit}. 
Thus the PL region extends to the whole open quarter-wave region with a quantified constant; only \emph{convexity} is confined to $\Treg{1/2}$.
\end{corollary}

\begin{proof}
Identical to the proof of Theorem~\ref{thm:whitebox-structure}: Lemma~\ref{s:lem:star} (now for
$\theta\in(0,1)$) gives
$\inner{\nabla G}{\vdelta}\ge\tfrac{2\sin(\theta\pi)}{\theta\pi}G$;
Cauchy--Schwarz and quadratic growth (Corollary~\ref{cor:sandwich}, valid on
$\Treg{1}\supseteq\Treg{\theta}$) give
$\norm{\nabla G}_2\ge\tfrac{2\sin(\theta\pi)}{\theta\pi}\,
\tfrac{\sqrt{2\lambda_r}}{\pi}\sqrt{G}$; square.
\end{proof}

\section{Proof of Theorem~\ref{thm:whitebox-gd} (white-box: Gradient Descent)}\label{s:gd}

\begin{lemma}[Hypothesis bridge]\label{s:lem:bridge}
If $\norm{\Proj\vdelta^0}_2\le\tfrac{1}{8\sqrt{\kappa_r}}$, then
$\norm{\Proj\vdelta^0}_2\le\tfrac{\pi}{16}$ and
$G(\vdelta^0)\le\tfrac{\lambda_r}{128}$.
\end{lemma}
\begin{proof}
Since $\kappa_r\ge1$,
$\tfrac{1}{8\sqrt{\kappa_r}}\le\tfrac18<\tfrac{\pi}{16}$. By
Corollary~\ref{cor:dom} and \eqref{s:QH},
$G(\vdelta^0)\le Q(\vdelta^0)\le\tfrac{\lambda_1}{2}
\norm{\Proj\vdelta^0}_2^2\le\tfrac{\lambda_1}{2}\cdot
\tfrac{1}{64\kappa_r}=\tfrac{\lambda_r}{128}$.
\end{proof}

\begin{restated}{Theorem~\ref{thm:whitebox-gd} (restated)}
Let $\vdelta^{t+1}=\vdelta^t-\lambda_1^{-1}\nabla G(\vdelta^t)$ with
$\norm{\Proj\vdelta^0}_2\le\tfrac{1}{8\sqrt{\kappa_r}}$. Then for all
$t$, $G(\vdelta^t)\le\bigl(1-\tfrac{16}{\pi^4\kappa_r}\bigr)^tG(\vdelta^0)$,
the iterates converge to a single point of the TL gauge orbit, and the gauge coordinates are constants of the motion.
\end{restated}

\begin{proof}
By Lemma~\ref{s:lem:bridge} we may assume
$\norm{\Proj\vdelta^0}_2\le\pi/16$ and
$G(\vdelta^0)\le\lambda_r/128$. Write
$\mu:=16\lambda_r/\pi^4$, so $\mu/\lambda_1=16/(\pi^4\kappa_r)$. We
prove by induction the invariant $\mathcal{I}_t$:
$\norm{\Proj\vdelta^t}_2\le\pi/16$ and $G(\vdelta^t)\le G(\vdelta^0)$.
The base case is the hypothesis. Assume $\mathcal{I}_t$.

\emph{One-step containment.} By Lemma~\ref{s:lem:smooth},
$\norm{\vdelta^{t+1}-\vdelta^t}_2
=\lambda_1^{-1}\norm{\nabla G(\vdelta^t)}_2
\le\norm{\Proj\vdelta^t}_2$, so
$\norm{\Proj\vdelta^{t+1}}_2\le2\norm{\Proj\vdelta^t}_2\le\pi/8$; since
$\norm{\Proj\vdelta}_\infty\le\norm{\Proj\vdelta}_2$, both $\vdelta^t$ and
$\vdelta^{t+1}$ lie in $\Treg{1/2}\subset\Treg{1}$ (the gauge component
is inert throughout by Lemma~\ref{s:lem:deriv}).

\emph{Descent and contraction.} Global $\lambda_1$-smoothness
(Lemma~\ref{s:lem:smooth}) gives the descent lemma
$G(\vdelta^{t+1})\le G(\vdelta^t)-\tfrac{1}{2\lambda_1}
\norm{\nabla G(\vdelta^t)}_2^2$; the PL inequality of Theorem~\ref{thm:whitebox-structure} applies
at $\vdelta^t\in\Treg{1/2}$, whence
$G(\vdelta^{t+1})\le(1-\mu/\lambda_1)G(\vdelta^t)$.

\emph{Closing the invariant.} $\vdelta^{t+1}\in\Treg{1}$, so quadratic
growth (Corollary~\ref{cor:sandwich}) gives
$\norm{\Proj\vdelta^{t+1}}_2^2\le\tfrac{\pi^2}{2\lambda_r}
G(\vdelta^{t+1})\le\tfrac{\pi^2}{2\lambda_r}\cdot
\tfrac{\lambda_r}{128}=\tfrac{\pi^2}{256}$, i.e.\
$\norm{\Proj\vdelta^{t+1}}_2\le\pi/16$; and
$G(\vdelta^{t+1})\le G(\vdelta^t)\le G(\vdelta^0)$. This establishes
$\mathcal{I}_{t+1}$ and the rate.

\emph{Convergence of the iterates.} The step norms
$\norm{\vdelta^{t+1}-\vdelta^t}_2\le\norm{\Proj\vdelta^t}_2\le
\tfrac{\pi}{\sqrt{2\lambda_r}}\sqrt{G(\vdelta^t)}$ decay geometrically,
so the sequence is Cauchy; its limit $\vdelta^\infty$ satisfies
$G(\vdelta^\infty)=0$ with $\Proj\vdelta^\infty=0$, and the gauge
component never moved (Lemma~\ref{s:lem:deriv}).
\end{proof}


\section{Proofs of Lemma~\ref{lem:barrier}, Theorem~\ref{thm:blackbox} and Proposition~\ref{prop:local} (black-box settings)}
\label{s:blackbox}

Throughout, $V:=\Gauge^\perp=\operatorname{range}H$ (so $\dim V=r$), $g:=G|_V$, balls $\mathcal B$ (open) and $\bar{\mathcal B}$ (closed) are Euclidean balls in $V$, $\vw$ is a standard normal vector on $V$, and $\omega_r$ is the volume of the unit ball of $\R^r$. The algorithm is that of Section~\ref{sec:linear_convergence}, i.e., Algorithm~2.1 of \citet{AAGM2022} with identity covariance ($\kappa=1$); their result numbers refer to arXiv:2009.08647v2. The \emph{gauge-quotient norm} $\dq(\vx):=\min_{\vg\in\Gauge}\norm{\vx-\vg}_\infty$ is a seminorm on $\R^N$ with kernel $\Gauge$ (a norm on $V$), and $\dq(\vx)\le\norm{\vx}_\infty\le\norm{\vx}$ for $\vx\in V$. The regions of \eqref{s:regions} are $\Treg{\theta}=\{\dq\le\theta\pi/4\}$, and $K:=\Treg{1}\cap V$ is a compact convex neighbourhood of $0$ in $V$ with interior $\{\dq<\pi/4\}$. At an interior index $i$ of the support window $W$, second differences obey $|x_{i-1}-2x_i+x_{i+1}|\le4\norm{\vx}_\infty$ and vanish on $\Gauge$.

Theorem~\ref{thm:blackbox} follows from Lemma~\ref{s:lem:aagm}, the analysis of \citet{AAGM2022} with their single-ball Assumption~A1 weakened to one ball per basin; the other lemmas supply its hypotheses.

\subsection{Reduction to the gauge quotient and to \citet{AAGM2022}}

\begin{lemma}[Exact projection]\label{s:lem:proj}
Let $(\vdelta^t,\sigma_t)$ be the $(1{+}1)$-ES on $G$ with mutations $\vz_t$. Then $(P\vdelta^t,\sigma_t)$ is the $(1{+}1)$-ES on $g$ with mutations $\vw_t:=P\vz_t$ and the same acceptance events, and $G(\vdelta^t)=g(P\vdelta^t)$. The gauge component, moved by $\sigma_t(I-P)\vz_t$ on acceptance, is a martingale that selection never sees.
\end{lemma}

\begin{proof}
Every $\va_k$ annihilates $\Gauge$, so $G=g\circ P$, and accepting $\vdelta^t+\sigma_t\vz_t$ for $G$ is accepting $P\vdelta^t+\sigma_t\vw_t$ for $g$; the step size depends only on this event. $P\vz_t$ and $(I-P)\vz_t$ are independent as orthogonal projections of a standard normal vector, so $(I-P)\vz_t$ is independent of acceptance given the past.
\end{proof}

\begin{lemma}[Per-basin variant of {\citealp[Sect.~4]{AAGM2022}}]\label{s:lem:aagm}
Let $f:V\to[0,\infty)$ be measurable, $\mathcal X:=\{0<f\le f_{\max}\}$, and $p^{\mathrm{succ}}(\bar\sigma;\vx):=\Pr[f(\vx+\bar\sigma f(\vx)\vw)\le f(\vx)]$. Let $C_\ell,C_u>0$ and $0\le p_{\mathrm{jump}}<\ptgt<p^{\mathrm{limit}}\le1$, and suppose each $\vx\in\mathcal X$ has a centre $\vmu(\vx)\in V$ and a basin $\mathcal S(\vx)$ with $\vx\in\mathcal S(\vx)\subseteq\{f\le f(\vx)\}$ such that
\begin{itemize}
\item[(H1)] $\mathcal B\bigl(\vmu(\vx),C_\ell\,y\bigr)\subseteq\{f<y\}$ for every $y\in(0,f(\vx)]$; \hfill(inner balls)
\item[(H2)] $\mathcal S(\vx)\subseteq\bar{\mathcal B}\bigl(\vmu(\vx),C_u f(\vx)\bigr)$; \hfill(outer ball of the own basin)
\item[(H3)] $\Pr\bigl[\vx+\sigma\vw\in\{f\le f(\vx)\}\setminus\mathcal S(\vx)\bigr]\le p_{\mathrm{jump}}$ for every $\sigma>0$; \hfill(rare jumps)
\item[(H4)] $\liminf_{\bar\sigma\downarrow0}\inf_{\vx\in\mathcal X}p^{\mathrm{succ}}(\bar\sigma;\vx)\ge p^{\mathrm{limit}}$. \hfill(uniform improvability)
\end{itemize}
Then there are $\nu\in(0,1)$ and $A',B>0$, depending only on $r$, $(\aup,\adown)$ and the data of (H1)--(H4), including the rate of the limit in (H4), such that the ES on $f$ started at $\vx_0\in\mathcal X$, $\sigma_0>0$ satisfies, with $\tau_\eta:=\min\{t:\ f(\vx_t)\le\eta\}$,
\begin{equation}\label{s:eq:aagm}
\E[\tau_\eta]\le\frac{\ln\frac{f(\vx_0)}{\eta}+\nu\bigl|\ln\frac{\sigma_0}{f(\vx_0)}\bigr|+A'}{B}\quad\bigl(0<\eta\le f(\vx_0)\bigr),
\qquad
\limsup_{t\to\infty}\tfrac1t\ln f(\vx_t)\le-B\ \ \text{a.s.}
\end{equation}
\end{lemma}

\begin{proof}
For $p_{\mathrm{jump}}=0$ and a single basin, this is Theorem~4.5 and Proposition~4.8 of \citet{AAGM2022} under their A1 with $a=0$, i.e.\ (H1)--(H2), and A2, i.e.\ (H4). A1 enters their Section~4 only through Proposition~4.1 (bounds on $p^{\mathrm{succ}}$) and Lemma~4.2 (step-size thresholds, and a lower bound on the probability $p^*_{r'}$ of $f(\vx+\sigma\vw)\le(1-r')f(\vx)$), which we recheck.

\emph{Lower bounds.} Their lower bounds on $p^{\mathrm{succ}}$ and $p^*_{r'}$ use only an inner ball around some $\vmu$ inside the target sublevel set and $\norm{\vx-\vmu}\le C_uf(\vx)$; (H1) at $y=(1-r')f(\vx)$ and (H2) supply both with $\vmu=\vmu(\vx)$, and the rest of the sublevel set only enlarges the success event.

\emph{Upper bound.} Their bound $p^{\mathrm{upper}}(\bar\sigma):=\sup_{\mathcal X}p^{\mathrm{succ}}(\bar\sigma;\cdot)\le\Pr[\norm{\vw}\le C_u/\bar\sigma]$ comes from the outer ball and Anderson's inequality \citep{Anderson1955} (a centred ball carries at least the Gaussian mass of any translate). Applying it to $\mathcal S(\vx)$ via (H2), and (H3) to the rest, gives $p^{\mathrm{upper}}(\bar\sigma)\le\Pr[\norm{\vw}\le C_u/\bar\sigma]+p_{\mathrm{jump}}\to p_{\mathrm{jump}}<\ptgt$ as $\bar\sigma\to\infty$.

\emph{Thresholds.} Fix $p_{\mathrm{hi}}\in(\max\{p_{\mathrm{jump}},2\ptgt-p^{\mathrm{limit}}\},\ptgt)$ and $p_{\mathrm{lo}}:=2\ptgt-p_{\mathrm{hi}}\in(\ptgt,p^{\mathrm{limit}})$ (their $p_u,p_\ell$). By the upper bound, $p^{\mathrm{succ}}\le p_{\mathrm{hi}}$ above a finite threshold of $\bar\sigma$; by (H4), $p^{\mathrm{succ}}\ge p_{\mathrm{lo}}$ below a positive one. From these, their Lemma~4.2 builds $s_{\mathrm{lo}}<s_{\mathrm{hi}}$ with $s_{\mathrm{hi}}/s_{\mathrm{lo}}\ge\aup/\adown$ (their $\ell,u$), and the lower bounds keep $p^*_{r'}$ away from $0$ on $[s_{\mathrm{lo}},s_{\mathrm{hi}}]$.

\emph{Drift.} Their Proposition~4.3 (Appendix~B.5) bounds the drift of
\[
\mathcal V(\vx,\sigma)=\ln f(\vx)+\nu\max\Bigl\{0,\ \ln\frac{\aup s_{\mathrm{lo}}f(\vx)}{\sigma},\ \ln\frac{\sigma}{\adown s_{\mathrm{hi}}f(\vx)}\Bigr\}
\]
using the objective only through these thresholds and elitism, which keeps $\vx_t\in\mathcal X$; their Lemma~4.4 uses neither A1 nor A2. Their Theorem~4.5 thus gives $\E[\tau_\eta]\le(\mathcal V(\vx_0,\sigma_0)-\ln\eta+A)/B$ with $A>B>0$, and the first identity of their Proposition~4.8 gives the a.s.\ rate (the second, on distances, needs A1 and is not used). Bounding the maximum in $\mathcal V(\vx_0,\sigma_0)$ by $|\ln(\sigma_0/f(\vx_0))|+C_{\mathcal V}$, $C_{\mathcal V}:=\max\{|\ln(\aup s_{\mathrm{lo}})|,|\ln(\adown s_{\mathrm{hi}})|\}$, gives \eqref{s:eq:aagm} with $A':=A+\nu C_{\mathcal V}$.
\end{proof}

\subsection{The wrap lattice, the barrier, and the basins}

\begin{lemma}[Wrap lattice]\label{s:lem:lattice}
$\Lat:=P(2\pi\Z^N)$ is a lattice in $V$ (a discrete subgroup spanning $V$). The function $g$ is $\Lat$-periodic with $g^{-1}(0)=\Lat$, and $\dq(\vmu)\ge\pi/2$, hence $\norm{\vmu}\ge\pi/2$, for every $\vmu\in\Lat\setminus\{0\}$. Moreover, $\operatorname{dist}(\vx,\Zset)=\min_{\vmu\in\Lat}\dq(\vx-\vmu)$ for $\vx\in V$, with $\operatorname{dist}$ as in \eqref{eq:barrier}.
\end{lemma}

\begin{proof}
$\Gauge=\ker H$ is spanned by integer vectors: the restrictions to $W$ of $\mathbf 1$ and of the ramp $(0,1,\dots,N-1)$, and $\ve_j$ for $j\notin W$ (Corollary~\ref{cor:gauge}, Appendix~\ref{s:setting}). For an integer basis matrix $U$ of $\Gauge$, $P=I-U(U^{\!\top}U)^{-1}U^{\!\top}$ is rational; with $q$ a common denominator, $\Lat\subseteq(2\pi/q)\Z^N$ is discrete, and it spans $V$ as the $P\ve_j$ do.
\emph{Periodicity.} $g(\vx+P(2\pi\vn))=G(\vx+2\pi\vn-(I-P)2\pi\vn)=G(\vx+2\pi\vn)=G(\vx)$, by gauge invariance and because $G$ depends on $\vdelta$ only through $e^{\mathrm i\vdelta}$.
\emph{Zeros.} By Corollary~\ref{cor:gauge}, $G^{-1}(0)=\Zset=2\pi\Z^N+\Gauge$, so $g^{-1}(0)=V\cap\Zset=P(V\cap\Zset)\subseteq P\Zset=\Lat\subseteq\Zset$.
\emph{Norm.} Let $\vmu=P(2\pi\vn)\neq0$ and $\vg\in\Gauge$. Then $\vmu-\vg=2\pi\vn-\vg'$ with $\vg'\in\Gauge$. As $\vmu\ne0$, $\vn$ is not affine on $W$, so at some interior index of $W$ the second difference of $\vn$ is a nonzero integer, and that of $\vmu-\vg$ is $2\pi$ times it. Hence $4\norm{\vmu-\vg}_\infty\ge2\pi$, and minimizing over $\vg$ gives $\dq(\vmu)\ge\pi/2$.
\emph{Distance.} $\operatorname{dist}(\vx,\Zset)=\min_{\vn,\vg}\norm{\vx-2\pi\vn-\vg}_\infty=\min_{\vn}\dq(\vx-2\pi\vn)=\min_{\vn}\dq(\vx-P(2\pi\vn))$, as $\dq$ is gauge invariant.
\end{proof}

\noindent
Let $F:=\{\vx\in V:\ \dq(\vx-\vmu)\ge\pi/4\ \ \forall\vmu\in\Lat\}$. By Lemma~\ref{s:lem:lattice} and gauge invariance, $\cstar=\inf_F g$. For $0\le c<\cstar$ define the \emph{basin} $S_c:=\{\vx\in K:\ g(\vx)\le c\}$.

\begin{lemma}[Barrier and basins; refines Lemma~\ref{lem:barrier}]\label{s:lem:basins}
\begin{enumerate}
\renewcommand{\labelenumi}{(\alph{enumi})}
\item $\cstar$ is attained and positive, $\cstar\le C/8$, and $\cstar\le G(\pi\ve_j)$ for every $j\in W$.
\item $\cstar\ge\min\{\lambda_r/8,\ G_{\mathrm{trap}}\}$.
\item For $c<\cstar$: $\{\vx\in V:\ g(\vx)\le c\}=\bigsqcup_{\vmu\in\Lat}(\vmu+S_c)$ (disjoint union); $S_c$ is compact, contained in the interior of $K$, and star-shaped with respect to $0$.
\item For $c<\cstar$: $\mathcal B\bigl(0,\sqrt{2c/\lambda_1}\bigr)\subseteq S_c\subseteq\bar{\mathcal B}\bigl(0,\pi\sqrt{c/(2\lambda_r)}\bigr)$, and $g<c$ on the open inner ball.
\item $S_0=\{0\}$, and $c\mapsto\operatorname{vol}(S_c)$ is continuous and strictly increasing on $[0,\cstar)$.
\end{enumerate}
\end{lemma}

\begin{proof}
(a) The map $\phi(\vx):=\min_{\vmu\in\Lat}\dq(\vx-\vmu)$ is continuous ($\dq$ is $1$-Lipschitz in $\norm{\cdot}_\infty$ and $\Lat$ is discrete), so $F=\{\phi\ge\pi/4\}$ is closed and $\Lat$-invariant. As $g$ is continuous and $\Lat$-periodic and $V/\Lat$ is compact, $\inf_F g$ is attained at some $\vx^\star\in F$ once $F\ne\emptyset$, and $\cstar=g(\vx^\star)>0$ as $F\cap\Lat=\emptyset$.
\emph{Comb.} Let $\vx:=P(\tfrac{\pi}{4}\vu)$. By Proposition~\ref{cx:crit}, $\dq(\vx)=\pi/4$, and for $\vmu\ne0$, $\dq(\vx-\vmu)\ge\dq(\vmu)-\dq(\vx)\ge\pi/4$ by Lemma~\ref{s:lem:lattice}. Hence $\vx\in F$ and $\cstar\le g(\vx)=C/8$.
\emph{Half-wrap.} Let $j\in W$ and pick an interior index $i$ of $W$ with $|i-j|=1$ (possible as $|W|\ge3$). For all $\vn\in\Z^N$ and $\vg\in\Gauge$, the second difference of $\pi\ve_j-2\pi\vn-\vg$ at $i$ is $\pi$ minus a multiple of $2\pi$, of modulus $\ge\pi$, so $\norm{\pi\ve_j-2\pi\vn-\vg}_\infty\ge\pi/4$. Thus $\operatorname{dist}(\pi\ve_j,\Zset)\ge\pi/4$, so $P(\pi\ve_j)\in F$ and $\cstar\le G(\pi\ve_j)$.

(b) If $\vx^\star$ is interior to $F$ (relative to $V$), it is a local minimizer of $g$, hence of $G=g\circ P$, with $G(\vx^\star)>0$, so $\cstar\ge G_{\mathrm{trap}}$. Otherwise $\phi(\vx^\star)=\pi/4$ by continuity of $\phi$, attained at some $\vmu$, so $\vy:=\vx^\star-\vmu\in K$ with $\dq(\vy)=\pi/4$. By Corollary~\ref{cor:sandwich} and $Q\ge\frac{\lambda_r}{2}\norm{\cdot}^2$ on $V$,
$\cstar=g(\vy)\ge\frac{4}{\pi^2}Q(\vy)\ge\frac{2\lambda_r}{\pi^2}\norm{\vy}^2\ge\frac{2\lambda_r}{\pi^2}\dq(\vy)^2=\frac{\lambda_r}{8}$.

(c) Let $c<\cstar$ and $g(\vx)\le c$. Then $\vx\notin F$, so $\vx-\vmu\in\operatorname{int}K$ for some $\vmu\in\Lat$. This $\vmu$ is unique, since two such $\vmu,\vmu'$ give $\dq(\vmu-\vmu')<\pi/2$ and hence $\vmu=\vmu'$ by Lemma~\ref{s:lem:lattice}; this is also the disjointness of the translates. By periodicity $\vx-\vmu\in S_c$, and conversely $\vmu+S_c\subseteq\{g\le c\}$. If $\dq(\vy)=\pi/4$, then $\dq(\vy-\vmu)\ge\pi/2-\pi/4$ for all $\vmu\ne0$, so $\vy\in F$ and $g(\vy)\ge\cstar>c$; hence $S_c\subseteq\operatorname{int}K$, and $S_c$ is compact as a closed subset of $K$. \emph{Star shape:} for $\vx\in K$, every $\ell_k=\inner{\va_k}{\vx}$ satisfies $|\ell_k|\le\pi/2$ (Corollary~\ref{cor:sandwich}), so $\tfrac{d}{d\tau}g(\tau\vx)=\sum_kw_k\,\ell_k\sin(2\tau\ell_k)\ge0$ for $\tau\in[0,1]$; hence $g(\tau\vx)\le g(\vx)$, and $\tau\vx\in K$ by convexity.

(d) \emph{Outer.} For $\vx\in S_c\subseteq K$, $c\ge g(\vx)\ge\frac{2\lambda_r}{\pi^2}\norm{\vx}^2$ as in (b). \emph{Inner.} If $\norm{\vx}<\sqrt{2c/\lambda_1}$, then $g(\vx)\le Q(\vx)\le\frac{\lambda_1}{2}\norm{\vx}^2<c$ by Corollary~\ref{cor:dom}. The open ball is connected, contains $0$, and lies in $\bigsqcup_{\vmu}(\vmu+S_c)$, whose pieces are compact, pairwise disjoint and locally finite, hence open and closed in the union; so the ball lies in $S_c$.

(e) $S_0=K\cap\Lat=\{0\}$ by Lemma~\ref{s:lem:lattice}. The level sets of the nonconstant trigonometric polynomial $g$ are Lebesgue-null, so $\operatorname{vol}(S_c)=\operatorname{vol}(K\cap\{g<c\})$, and continuity follows from continuity of measure along monotone families. For $0\le c<c'<\cstar$, $S_{c'}\setminus S_c$ contains the open set $\{\vx\in\operatorname{int}K:\ c<g(\vx)<c'\}$, which is nonempty: along a segment from $0$ to $\partial K\subseteq F$, $g$ increases continuously from $0$ to at least $\cstar>c'$.
\end{proof}

\noindent
\emph{Proof of Lemma~\ref{lem:barrier}.} Items (a)--(c) of Lemma~\ref{s:lem:basins}, lifted from $V$ to $\R^N$ by gauge invariance. \qed

\medskip
\noindent
\emph{Folded suboptimality.} For $c\in[0,\cstar)$ let $h(c):=\operatorname{vol}(S_c)^{1/r}$, which is continuous and strictly increasing with $h(0)=0$ by Lemma~\ref{s:lem:basins}(e). Extend $h$ strictly increasingly to $[0,\infty)$ and set $f:=h\circ g$. This is the \emph{spatial suboptimality} of \citet{AAGM2022}, computed on one basin instead of the whole periodic, infinite-volume sublevel set. By Lemma~\ref{s:lem:basins}(d),
\begin{equation}
\omega_r^{1/r}\sqrt{2c/\lambda_1}\;\le\;h(c)\;\le\;\omega_r^{1/r}\,\pi\sqrt{c/(2\lambda_r)} .
\label{s:hbounds}
\end{equation}
Let $C_\ell:=2/(\pi\,\omega_r^{1/r}\sqrt{\kappa_r})$ and $C_u:=\frac{\pi}{2}\,\omega_r^{-1/r}\sqrt{\kappa_r}$. Combining \eqref{s:hbounds} with Lemma~\ref{s:lem:basins}(d), periodicity, and $\frac{2\lambda_r}{\pi^2}\norm{\vx}^2\le g(\vx)\le\frac{\lambda_1}{2}\norm{\vx}^2$ on $K$ gives, for all $\vmu\in\Lat$ and $c<\cstar$,
\begin{equation}
\mathcal B\bigl(\vmu,C_\ell\,h(c)\bigr)\subseteq\{f<h(c)\},
\qquad
C_\ell f(\vx)\le\norm{\vx-\vmu}\le C_uf(\vx)\quad(\vx\in\vmu+S_c).
\label{s:fnorm}
\end{equation}

\subsection{Periodic images and uniform improvability}

\begin{lemma}[Periodic images]\label{s:lem:images}
Let $0<c\le\lambda_r/32$ with $c<\cstar$, and let $\vmu'\in\Lat$. For every $\vx\in\vmu'+S_c$ and every $\sigma>0$,
\[
\Pr\Bigl[\vx+\sigma\vw\in\textstyle\bigcup_{\vmu\in\Lat\setminus\{\vmu'\}}(\vmu+S_c)\Bigr]\;\le\;\Bigl(\frac{128\,c}{\lambda_r}\Bigr)^{r/2}.
\]
\end{lemma}

\begin{proof}
By periodicity let $\vmu'=0$. Let $R:=\pi\sqrt{c/(2\lambda_r)}\le\pi/8$ and $D:=\min_{\vmu\in\Lat\setminus\{0\}}\norm{\vmu}\ge\pi/2$ (Lemma~\ref{s:lem:lattice}), so $2R\le D/2$. By Lemma~\ref{s:lem:basins}(d), $\vx\in\bar{\mathcal B}(0,R)$ and $\vmu+S_c\subseteq\bar{\mathcal B}(\vmu,R)$, whose points $\vy$ satisfy $\norm{\vy-\vx}\ge\norm{\vmu}-2R\ge\norm{\vmu}/2$. Bounding the density of $\vx+\sigma\vw$ there,
\[
\Pr[\vx+\sigma\vw\in\vmu+S_c]\;\le\;\omega_rR^r\,(2\pi\sigma^2)^{-r/2}\,e^{-\norm{\vmu}^2/(8\sigma^2)} .
\]
The balls $\mathcal B(\vmu,D/2)$, $\vmu\in\Lat$, are pairwise disjoint, and $\norm{\vy}\le2\norm{\vmu}$ on $\mathcal B(\vmu,D/2)$ for $\vmu\neq0$, so $e^{-\norm{\vmu}^2/(8\sigma^2)}\le e^{-\norm{\vy}^2/(32\sigma^2)}$ there. Averaging over each ball and summing,
\[
\sum_{\vmu\ne0}e^{-\norm{\vmu}^2/(8\sigma^2)}
\;\le\;\frac{1}{\omega_r(D/2)^r}\int_{V}e^{-\norm{\vy}^2/(32\sigma^2)}\dd\vy
\;=\;\frac{(32\pi\sigma^2)^{r/2}}{\omega_r(D/2)^r}.
\]
Multiplying, the total is at most $R^r\,16^{r/2}(2/D)^r=(8R/D)^r\le(16R/\pi)^r=(128c/\lambda_r)^{r/2}$.
\end{proof}

\begin{lemma}[Uniform $\tfrac12$-improvability]\label{s:lem:improv}
Let $0<b<\cstar$ and $\theta_b:=\frac{4}{\pi}\max_{\vx\in S_b}\dq(\vx)$, so that $\theta_b<1$ by Lemma~\ref{s:lem:basins}(c). For every $p<\tfrac12$ there is $s_p>0$ such that $\Pr[g(\vx+\sigma\vw)<g(\vx)]\ge p$ whenever $0<g(\vx)\le b$ and $0<\sigma\le s_p\,f(\vx)$.
\end{lemma}

\begin{proof}
By periodicity let $\vx\in S_b\setminus\{0\}\subseteq\Treg{\theta_b}$. By Lemma~\ref{s:lem:deriv}, $\nabla G(\vx)\in V$, so $\nabla g=\nabla G$ on $V$; Lemma~\ref{s:lem:star} and $Q\ge\frac{\lambda_r}{2}\norm{\cdot}^2$ on $V$ give $\norm{\nabla g(\vx)}\ge\inner{\nabla G(\vx)}{\vx}/\norm{\vx}\ge\frac{\sin(\theta_b\pi)}{\theta_b\pi}\lambda_r\norm{\vx}$. By Taylor's formula and Lemma~\ref{s:lem:smooth}, $g(\vx+\sigma\vw)\le g(\vx)+\sigma\inner{\nabla g(\vx)}{\vw}+\frac{\sigma^2\lambda_1}{2}\norm{\vw}^2$. With $\hat{\ve}:=\nabla g(\vx)/\norm{\nabla g(\vx)}$, $s:=\sigma/f(\vx)$, $\norm{\vx}\ge C_\ell f(\vx)$ from \eqref{s:fnorm}, and $\kappa':=\theta_b\pi\kappa_r/(2C_\ell\sin(\theta_b\pi))$,
\[
\bigl\{g(\vx+\sigma\vw)<g(\vx)\bigr\}\;\supseteq\;\Bigl\{\inner{\hat{\ve}}{\vw}<-\frac{\sigma\lambda_1\norm{\vw}^2}{2\norm{\nabla g(\vx)}}\Bigr\}
\;\supseteq\;\bigl\{\inner{\hat{\ve}}{\vw}<-s\,\kappa'\norm{\vw}^2\bigr\}.
\]
By rotation invariance of $\vw$, the last event has a probability $q(s)$ that does not depend on $\vx$, is nonincreasing in $s$, and tends to $\Pr[\inner{\hat{\ve}}{\vw}<0]=\tfrac12$ as $s\downarrow0$ by monotone convergence. Choose $s_p$ with $q(s_p)\ge p$.
\end{proof}

\subsection{Proof of Theorem~\ref{thm:blackbox}}

The assumption on $b$ gives $b<\bstar\le\min\{\lambda_r/8,G_{\mathrm{trap}}\}\le\cstar$ (Lemma~\ref{s:lem:basins}(b)) and $\varepsilon_b:=(128b/\lambda_r)^{r/2}<\ptgt$, so in particular $b<\lambda_r/128$.

\emph{(i) Confinement.} By Lemma~\ref{s:lem:proj}, $\vx_t:=P\vdelta^t$ is the $(1{+}1)$-ES on $g$ with the same step sizes. By elitism, $g(\vx_t)\le g(\vx_0)\le b<\cstar$, so by Lemma~\ref{s:lem:basins}(c), $\vx_t\in\vmu_t+S_b$ for a unique $\vmu_t=P(2\pi\vn_t)\in\Lat$, and $\dq(\vdelta^t-2\pi\vn_t)=\dq(\vx_t-\vmu_t)<\pi/4$ by gauge invariance of $\dq$. Thus $\vdelta^t\in2\pi\Z^N+\Treg{1}$; the basin may change with~$t$.

\emph{(ii)--(iii) Rates.} The ES is comparison-based and $h$ is strictly increasing, so the ES on $g$ and on $f=h\circ g$ generate identical trajectories. Apply Lemma~\ref{s:lem:aagm} to $f$ with $f_{\max}:=h(b)$, $\vmu(\vx)$ the lattice point of the basin of $\vx$, and $\mathcal S(\vx):=\vmu(\vx)+S_{g(\vx)}$; by Lemma~\ref{s:lem:basins}(c), $\{f\le f(\vx)\}=\bigsqcup_{\vmu\in\Lat}(\vmu+S_{g(\vx)})$. Then (H1) and (H2) follow from \eqref{s:fnorm}, as $h$ maps $(0,g(\vx)]$ onto $(0,f(\vx)]$; (H3) holds with $p_{\mathrm{jump}}=\varepsilon_b$ by Lemma~\ref{s:lem:images} at $c=g(\vx)\le b$; and (H4) holds with $p^{\mathrm{limit}}=\tfrac12>\ptgt$ by Lemma~\ref{s:lem:improv}. The constants of Lemma~\ref{s:lem:aagm} thus depend only on $r,\kappa_r,\theta_b,\varepsilon_b$ and $(\aup,\adown)$.

By \eqref{s:hbounds}, $G(\vdelta^t)=g(\vx_t)\le\lambda_1f(\vx_t)^2/(2\omega_r^{2/r})$, so \eqref{s:eq:aagm} gives $\limsup_t\frac1t\ln G(\vdelta^t)\le-2B$ almost surely, which is (iii).
For (ii), let $\varepsilon<G(\vdelta^0)$ (otherwise $T_\varepsilon=0$) and $\eta:=\omega_r^{1/r}\sqrt{2\varepsilon/\lambda_1}$. Then $T_\varepsilon\le\tau_\eta$ and $\eta<f(\vx_0)$, so \eqref{s:eq:aagm} applies, with $\ln(f(\vx_0)/\eta)\le\frac12\ln(G(\vdelta^0)/\varepsilon)+\ln(\frac{\pi}{2}\sqrt{\kappa_r})$ by \eqref{s:hbounds} and $|\ln(\sigma_0/f(\vx_0))|\le|\ln(\sigma_0/\norm{\vx_0-\vmu_0})|+\max\{|\ln C_\ell|,|\ln C_u|\}$ by \eqref{s:fnorm}. Here $\norm{\vx_0-\vmu_0}=\min_{\vmu\in\Lat}\norm{\vx_0-\vmu}=\operatorname{dist}_2(\vdelta^0,\Zset)$, since $\norm{\vx_0-\vmu_0}<\pi/16$ by Lemma~\ref{s:lem:basins}(d) while distinct lattice points are at least $\pi/2$ apart (Lemma~\ref{s:lem:lattice}), and $\Zset=2\pi\Z^N+\Gauge$. This gives (ii) with
\[
C_1:=\frac{\nu}{B},\qquad C_0:=\frac{1}{B}\Bigl(A'+\ln\bigl(\tfrac{\pi}{2}\sqrt{\kappa_r}\bigr)+\nu\max\{|\ln C_\ell|,|\ln C_u|\}\Bigr).\qquad\qed
\]

\subsection{Proof of Proposition~\ref{prop:local}}

On $V$, since $0\le x^2-\sin^2x\le x^4/3$ and $\ell_k^2\le\norm{\va_k}^2\norm{\vx}^2\le\tfrac32\norm{\vx}^2$ (each $\va_k$ has at most four entries in $\{\pm\tfrac12,\pm1\}$, and $\norm{\va_k}^2\in\{0,1,\tfrac32\}$),
\[
0\;\le\;Q(\vx)-G(\vx)\;=\;\sum_kw_k\bigl(\ell_k^2-\sin^2\ell_k\bigr)\;\le\;\frac{\norm{\vx}^2}{2}\,Q(\vx)\;\le\;\frac{\lambda_1}{4}\norm{\vx}^4 .
\]
Fix $\hat{\vy}$ and $s$, and write $\Delta_Q:=Q(\hat{\vy}+s\vw)-Q(\hat{\vy})$ and $\Delta_G:=\rho^{-2}[G(\rho(\hat{\vy}+s\vw))-G(\rho\hat{\vy})]$; by homogeneity of $Q$, the two acceptance indicators are $\Ind{\Delta_Q\le0}$ and $\Ind{\Delta_G\le0}$. The display gives $|\Delta_G-\Delta_Q|\le\frac{\lambda_1\rho^2}{4}(\norm{\hat{\vy}+s\vw}^4+1)$, so for any $M>0$ the indicators differ with probability at most
\[
\Pr[\norm{\vw}>M]\;+\;\Pr\bigl[|\Delta_Q|\le\eta_\rho\bigr],\qquad \eta_\rho:=\tfrac{\lambda_1\rho^2}{4}\bigl((1+sM)^4+1\bigr).
\]
Now $\Delta_Q=s\inner{H\hat{\vy}}{\vw}+\frac{s^2}{2}\vw^{\!\top}\!H\vw$ is a nonconstant polynomial in $\vw$ ($H\hat{\vy}\neq0$ as $\hat{\vy}\in\operatorname{range}H$), so its law has no atoms and $u_\eta(\hat{\vy},s):=\Pr[|\Delta_Q|\le\eta]\downarrow0$ as $\eta\downarrow0$. Each $u_\eta$ is upper semicontinuous (portmanteau; $\Delta_Q$ is continuous in $(\hat{\vy},s,\vw)$), so by Dini's theorem the convergence is uniform on the compact set $\{\norm{\hat{\vy}}=1\}\times[s_1,s_2]$. Choosing $M$ large and then $\rho$ small proves the first claim. For the second, the next state $(\vy+\sigma\vw\,\Ind{\text{accept}},\,\sigma\aup^{\Ind{\text{accept}}}\adown^{1-\Ind{\text{accept}}})$ is the same function of $\vw$ under $G$ and $Q$ outside the disagreement event, so the total-variation distance between the two transition laws is at most its probability. \qed

\subsection{Sphere reference constants}

\begin{lemma}[Asymptotic progress rate on the Sphere \cite{Beyer2001}]\label{s:lem:rate}
Consider the $(1{+}1)$-ES on an isotropic quadratic in $\R^r$ at distance $\rho$ from the optimum, with normalized step $\sigma^\star:=\sigma r/\rho$. As $r\to\infty$, its normalized progress rate $\psi^\star:=r\,\E[(\rho-\rho')/\rho]$ and its success probability converge to
\begin{equation}
\psi^\star(\sigma^\star)=\sigma^\star\phi\bigl(\tfrac{\sigma^\star}{2}\bigr)-\frac{\sigma^{\star2}}{2}\Bigl(1-\Phi\bigl(\tfrac{\sigma^\star}{2}\bigr)\Bigr),
\qquad
P_{\mathrm{succ}}=1-\Phi\bigl(\tfrac{\sigma^\star}{2}\bigr),
\label{s:phistar}
\end{equation}
where $\phi,\Phi$ are the standard normal density and distribution function. $\psi^\star$ is maximal at $\sigma^\star_{\mathrm{opt}}=1.2240$, with $\psi^\star_{\max}=0.2025$ and $P_{\mathrm{succ}}=0.2703$; at $P_{\mathrm{succ}}=1/5$, $\sigma^\star_{1/5}=2\Phi^{-1}(0.8)=1.6832$ and $\psi^\star_{1/5}=0.1879$. The resulting log-rate of the objective $\propto\rho^2$ is $2\psi^\star/r$.
\end{lemma}

\noindent These are the reference constants of Section~\ref{sec:numerical}. The progress rate $\psi^\star$ is written $\varphi^\star$ in the ES literature; we rename it because $\varphi$ is the spectral phase.

\section{Sharpness of the trust radii: the comb witness}\label{s:comb}

\begin{proposition}[Exact convexity region]\label{cx:prop}
$G$ is convex on the convex set $\Treg{1/2}$, and
$\cos(\theta\pi)\lambda_r$-strongly convex modulo gauge on
$\Treg{\theta}$ for $\theta<\tfrac12$; it is nowhere strictly convex
($G$ is constant along the two-dimensional gauge $\Gauge$). The radius
is exact: along the alternating comb $\vu=((-1)^i)_{i}$,
\begin{equation}
\frac{d^2}{dc^2}\,G(c\,\vu)\;=\;C\cos(4c),
\qquad
C=\vu^{\!\top}H\vu>0,
\label{cx:comb}
\end{equation}
which is strictly negative for $c\in(\pi/8,\,3\pi/8)$. Hence
$\Treg{1/2}$ is the largest ball of the family
$\{\Treg{\theta}\}$ on which $G$ is convex, for the SHG instance
itself and not merely for the class.
\end{proposition}

\begin{proof}
Convexity and strong convexity are the termwise Hessian bounds
\[\nabla^2G(\vdelta)=\sum_k2w_k\cos(2\inner{\va_k}{\vdelta})
\va_k^{\vphantom{\top}}\va_k^{\!\top}
\succeq\cos(\theta\pi)H
\]
on $\Treg{\theta}$, $\theta\le\tfrac12$ (Theorem~\ref{thm:whitebox-structure}).
For the comb: $\inner{\va_k}{\vu}
=\tfrac12[(-1)^i+(-1)^{s-i}-(-1)^p-(-1)^{s-p}]$ vanishes on odd
shells and on same-parity pairs, and equals $\pm2$ on parity-mixed
pairs of even shells; there $\inner{\va_k}{c\vu}=\pm2c$. Hence
\[
\frac{d^2}{dc^2}G(c\vu)
=\sum_k2w_k\cos\!\bigl(2\inner{\va_k}{c\vu}\bigr)\inner{\va_k}{\vu}^2
=\cos(4c)\!\!\sum_{k\ \mathrm{mixed}}\!\!8w_k
=C\cos(4c),
\]
with $C=\vu^{\!\top}H\vu=\sum_k2w_k\inner{\va_k}{\vu}^2>0$: the
hypothesis $|W|\ge3$ supplies an even shell carrying a parity-mixed
pair of positive weight, hence at least one active term. For
$c\in(\pi/8,3\pi/8)$ the directional curvature is negative, so
$\nabla^2G(c\vu)\not\succeq0$: no $\Treg{\theta}$ with
$\theta>\tfrac12$ is a convexity region.
\end{proof}

\begin{proposition}[Critical point on the quarter-wave boundary]
\label{cx:crit}
Let $\vdelta=\tfrac{\pi}{4}\vu$ with $\vu=((-1)^i)_i$ and $|W|\ge3$. Then
$\vdelta\in\partial\Treg{1}$, $\nabla G(\vdelta)=\mathbf 0$, and
$G(\vdelta)=C/8>0$ with $C=\vu^{\!\top}H\vu$ as in \eqref{cx:comb}. Hence no
star-convexity or PL constant survives at $\theta=1$ --- for the Quantum
Sphere itself, not merely for the class. 
\end{proposition}

\begin{proof}
By the parity computation of Proposition~\ref{cx:prop},
$\inner{\va_k}{\vu}$ vanishes on odd shells and same-parity pairs and equals
$\pm2$ on parity-mixed pairs of even shells; hence
$\inner{\va_k}{\vdelta}\in\{0,\pm\pi/2\}$ for every $k$, so
$\sin\bigl(2\inner{\va_k}{\vdelta}\bigr)=0$ throughout and
$\nabla G(\vdelta)=\mathbf 0$ by Lemma~\ref{s:lem:deriv}, while
$G(\vdelta)=\sum_{k\ \mathrm{mixed}}w_k=C/8>0$.
For the boundary claim, let $g=\alpha\mathbf 1+\beta\,(0,\dots,N-1)$ be any
gauge element and $M=\norm{\vdelta-g}_\infty$. Consecutive differences of
$\vdelta$ alternate between $+\pi/2$ and $-\pi/2$, so those of $\vdelta-g$
alternate between $\pi/2+\beta$ and $-\pi/2+\beta$; each is at most $2M$ in
modulus, and adding the two bounds gives
$\pi\le|\pi/2+\beta|+|\pi/2-\beta|\le4M$, i.e.\ $M\ge\pi/4$, attained at
$g=0$. Thus $\min_{g\in\Gauge}\norm{\vdelta-g}_\infty=\pi/4$ exactly.
\end{proof}

\begin{remark}[The comb as universal witness]\label{cx:rem:comb}
The same configuration family $\vdelta=c\,((-1)^i)$ certifies three
sharpness statements: at $c=\pi/4$ it attains the sandwich constant $G/Q=4/\pi^2$ exactly and is a critical point with positive gap (Proposition~\ref{cx:crit}); 
at $c=\pi/2$ it is a gauge configuration (slope-$\pi$ affine phase mod $2\pi$); and by \eqref{cx:comb} it pins the convexity radius at $c=\pi/8$ --- the three radii in ratio
$1:2:4$, all read off one direction. Numerically the comb is extremal:
sampled boundary points of larger balls remain positive semidefinite
well past $\pi/8$, so the full set $\{\nabla^2G\succeq0\}$ properly
contains $\Treg{1/2}$ but contains no larger ball.
\end{remark}

\newpage

\section{Explicit gradient and Hessian at arbitrary phase}
\label{app:derivatives}

Throughout, $m,n\in\{0,\dots,N-1\}$ index the pulse-shaper pixels, and we
retain the per-shell notation $\alpha^{(s)}_i=A_iA_{s-i}$,
$v^{(s)}_i=\delta_i+\delta_{s-i}$ of Corollary~\ref{cor:sin2-disc}.
Define the \emph{shell amplitude}
\begin{equation}
F_s(\vdelta) \;:=\; \sum_{i\in S_s} \alpha^{(s)}_i\,
e^{\,\mathrm{i}v^{(s)}_i},
\qquad Y(\vdelta)=\sum_{s=0}^{2N-2}\bigl|F_s(\vdelta)\bigr|^{2},
\label{eq:shell-amp}
\end{equation}
and the symmetric complex array
\begin{equation}
Z_{mn}(\vdelta) \;:=\; A_mA_n\,e^{\,\mathrm{i}(\delta_m+\delta_n)}\,
\cj{F_{m+n}(\vdelta)} ,
\label{eq:Zmn}
\end{equation}
which is well defined since $m+n\le 2N-2$. The conjugation in
\eqref{eq:Zmn} is essential: replacing $\cj{F_s}$ by $F_s$ in
\eqref{eq:grad} does not reproduce the real form \eqref{eq:grad-real}.
Since $G=Y(0)-Y$, all derivatives of $G$ are those of $Y$ with the
opposite sign.

\subsection{Gradient}
\begin{proposition}[Gradient]\label{prop:grad}
For every $\vdelta\in\R^{N}$ and every $m$,
\begin{equation}
\frac{\partial G}{\partial \delta_m}
 \;=\; 4\sum_{s:\,m\in S_s}\alpha^{(s)}_m\,
 \Im\!\left(e^{\,\mathrm{i}v^{(s)}_m}\,\cj{F_s}\right)
 \;=\; 4\sum_{n=0}^{N-1}\Im\, Z_{mn} .
\label{eq:grad}
\end{equation}
Equivalently, in purely real form,
\begin{equation}
\frac{\partial G}{\partial \delta_m}
 = 4\!\!\sum_{s:\,m\in S_s}\!\!\alpha^{(s)}_m
 \sum_{p\in S_s}\alpha^{(s)}_p\,
 \sin\!\bigl(v^{(s)}_m-v^{(s)}_p\bigr).
\label{eq:grad-real}
\end{equation}
\end{proposition}

\begin{proof}
Differentiating \eqref{eq:shell-amp} gives
$\partial_m F_s = 2\mathrm{i}\,\alpha^{(s)}_m e^{\mathrm{i}v^{(s)}_m}$ for
$m\in S_s$ (the two incidences $i=m$ and $i=s-m$ contribute equally, by
$\alpha^{(s)}_{s-m}=\alpha^{(s)}_m$ and $v^{(s)}_{s-m}=v^{(s)}_m$; for
$s=2m$ they coincide and the factor $2$ persists). Hence
$\partial_m|F_s|^2 = 2\Re(\cj{F_s}\,\partial_mF_s)
= -4\alpha^{(s)}_m\Im(e^{\mathrm{i}v^{(s)}_m}\cj{F_s})$, and
$\partial_m G=-\partial_m Y$ yields \eqref{eq:grad}. Reindexing $s=m+n$
gives the $Z$-form; expanding $\Im(e^{\mathrm{i}v_m}\cj{F_s})$
gives \eqref{eq:grad-real}.
\end{proof}

\subsection{Hessian}
\begin{proposition}[Hessian]\label{prop:hess}
Let
\begin{equation}
C_{mn}(\vdelta) := \!\!\sum_{s:\,m,n\in S_s}\!\!
 \alpha^{(s)}_m\alpha^{(s)}_n
 \cos\!\bigl(v^{(s)}_m-v^{(s)}_n\bigr),
\label{eq:Cmn}
\end{equation}
the sum running over $\max(m,n)\le s\le\min(m,n)+N-1$. Then for every
$\vdelta\in\R^{N}$,
\begin{equation}
H_{mn}(\vdelta)=\frac{\partial^{2}G}{\partial\delta_m\partial\delta_n}
 = -8\,C_{mn} + 4\,\Re Z_{mn}
 + 4\,\Ind{m=n}\sum_{k=0}^{N-1}\Re Z_{mk}.
\label{eq:hess-general}
\end{equation}
Both $C$ and $Z$ are symmetric, so $H$ is symmetric.
\end{proposition}

\begin{proof}
From $\partial_m F_s=2\mathrm{i}\alpha^{(s)}_m e^{\mathrm{i}v^{(s)}_m}$,
\[
\partial_n\partial_m F_s
 = -2\alpha^{(s)}_m e^{\mathrm{i}v^{(s)}_m}
 \bigl(\Ind{n=m}+\Ind{n=s-m}\bigr).
\]
Then $\partial_n\partial_m|F_s|^2
 = 2\Re\bigl(\cj{\partial_nF_s}\,\partial_mF_s
 +\cj{F_s}\,\partial_n\partial_mF_s\bigr)$. The first term gives
$8\alpha^{(s)}_m\alpha^{(s)}_n\cos(v^{(s)}_m-v^{(s)}_n)$; the second gives
$-4\alpha^{(s)}_m\Re(e^{\mathrm{i}v^{(s)}_m}\cj{F_s})
 (\Ind{n=m}+\Ind{n=s-m})$, where the indicator $\Ind{n=s-m}$
selects the single shell $s=m+n$ and contributes $-4\Re Z_{mn}$, while
$\Ind{n=m}$ contributes the diagonal sum. Negating
($H=-\nabla^2Y$) yields \eqref{eq:hess-general}.
\end{proof}

\begin{remark}[Structural rank-one form]\label{rem:rank1}
Writing $G=\sum_s\sum_{i,p\in S_s}\alpha^{(s)}_i\alpha^{(s)}_p
(1-\cos\theta^{(s)}_{ip})$ with
$\theta^{(s)}_{ip}=v^{(s)}_i-v^{(s)}_p$ and incidence vectors
$\bigl(\vc^{(s)}_{ip}\bigr)_m
 =\Ind{m=i}+\Ind{m=s-i}-\Ind{m=p}-\Ind{m=s-p}$, one obtains the
equivalent coordinate-free form
\begin{equation}
H(\vdelta)=\sum_{s}\sum_{i,p\in S_s}
 \alpha^{(s)}_i\alpha^{(s)}_p\,\cos\theta^{(s)}_{ip}\;
 \vc^{(s)}_{ip}\,\vc^{(s)\top}_{ip},
\label{eq:hess-rank1}
\end{equation}
which exhibits $H$ as a weighted sum of rank-one forms. At $\vdelta=0$
every $\cos\theta^{(s)}_{ip}=1$, so $H(0)\succeq0$ and
$\vu\in\ker H(0)$ iff $\vc^{(s)\top}_{ip}\vu=0$ for all positively
weighted $(s,i,p)$ --- the affine phase space of Corollary~\ref{cor:gauge}.
\end{remark}

\begin{remark}[Consistency at the TL]\label{rem:focal-check}
At $\vdelta=0$ all $F_s$ are real and positive, and
\eqref{eq:hess-general} reduces to
\begin{equation}
H_{mn}(0) = -8A_mA_n\!\!\sum_{s}\!A_{s-m}A_{s-n}
 + 4A_mA_nF_{m+n}
 + 4\,\Ind{m=n}\,T_m,
\end{equation}
with $T_m=A_m\sum_{s}A_{s-m}F_s$. Negating gives $\nabla^2Y(0)$, whose
three terms match, one for one and with matching signs $(-4,-4,+8)$, the
diagonal, sum-argument, and difference-argument terms of the analytic SHG Hessian, confirming $H=-H_{\mathrm{SHG}}$.
\end{remark}

\begin{remark}[Continuum analogue and computational cost]\label{rem:cont-cost}
Replacing each per-shell sum by $\int\dd\omega'$ against
$\mu_\omega(\omega')$ and $\sum_s$ by $\int\dd\omega$ turns
\eqref{eq:grad}--\eqref{eq:hess-general} into their continuum
counterparts, with $F_s\mapsto E_2(\omega)$.
Computationally, $F_s=(g*g)_s$ is the self-convolution of
$g_i=A_ie^{\mathrm{i}\delta_i}$, so $Y$, $\nabla G$, and any
Hessian-vector product are obtainable in $\mathcal{O}(N\log N)$ by FFT.
The dense Hessian costs $\mathcal{O}(N^2)$ \emph{in total} (rather than $\mathcal{O}(N^3)$): since $v^{(s)}_m-v^{(s)}_n=(\delta_m-\delta_n)+(\delta_{s-m}-\delta_{s-n})$, the
shell sum in \eqref{eq:Cmn} is a lag-$(m-n)$ autocorrelation,
\begin{equation}
C_{mn}(\vdelta)\;=\;\Re\Bigl[\,g_m\,\cj{g_n}\;R(m-n)\Bigr],
\qquad
R(\tau)\;:=\;\sum_{u} g_{u+\tau}\,\cj{g_u},
\label{eq:C-toeplitz}
\end{equation}
with $R$ obtained by a single FFT and each entry then costing
$\mathcal{O}(1)$.\\ 
Equation \eqref{eq:C-toeplitz} also exhibits the structure: the shell core of $H$ is Toeplitz, conjugated by $\operatorname{diag}(\vg)$.
\end{remark}

\end{document}